\documentclass[sigconf, nonacm]{acmart}

\newcommand\vldbdoi{10.14778/3587136.3587143}
\newcommand\vldbpages{1685 - 1698}
\newcommand\vldbvolume{16}
\newcommand\vldbissue{7}
\newcommand\vldbyear{2023}
\newcommand\vldbauthors{\authors}
\newcommand\vldbtitle{\shorttitle} 
\newcommand\vldbavailabilityurl{}
\newcommand\vldbpagestyle{empty} 

\usepackage[a-2b]{pdfx}
\usepackage{subfig}
\usepackage{graphicx}
\usepackage{amsthm, amsfonts}
\usepackage[linesnumbered,ruled,vlined]{algorithm2e}

\usepackage{pifont}

\newtheorem{definition}{Definition}
\newtheorem{theorem}{Theorem}

\usepackage{hyperref}

\def\BibTeX{{\rm B\kern-.05em{\sc i\kern-.025em b}\kern-.08em
    T\kern-.1667em\lower.7ex\hbox{E}\kern-.125emX}}

\begin{document}

\title{\textsc{GriDB}: Scaling Blockchain Database via Sharding and Off-Chain Cross-Shard Mechanism}

\author{Zicong Hong}
\affiliation{
\institution{Hong Kong Polytechnic University}}
\email{zicong.hong@connect.polyu.hk}

\author{Song Guo}
\affiliation{
\institution{Hong Kong Polytechnic University, PolyU Shenzhen Research Institute}}
\email{song.guo@polyu.edu.hk}

\author{Enyuan Zhou}
\affiliation{
\institution{Hong Kong Polytechnic University}}
\email{21038299r@connect.polyu.hk}

\author{Wuhui Chen}
\affiliation{
\institution{Sun Yat-sen University}}
\email{chenwuh@mail.sysu.edu.cn}

\author{Huawei Huang}
\affiliation{
\institution{Sun Yat-sen University}}
\email{huanghw28@mail.sysu.edu.cn}

\author{Albert Zomaya}
\affiliation{
\institution{The University of Sydney}}
\email{albert.zomaya@sydney.edu.au}

% \author{Zicong Hong$^{1}$, Song Guo$^{1,2}$, Enyuan Zhou$^{1}$, Wuhui Chen$^{3}$, Huawei Huang$^{3}$, and Albert Zomaya$^{4}$}
% \affiliation{
% \institution{$^1$The Hong Kong Polytechnic University, $^2$The Hong Kong Polytechnic University Shenzhen Research Institute, $^3$Sun Yat-sen University, $^4$The University of Sydney}
% }
% \email{zicong.hong@connect.polyu.hk, song.guo@polyu.edu.hk, 21038299r@connect.polyu.hk}
% \email{chenwuh@mail.sysu.edu.cn, huanghw28@mail.sysu.edu.cn, albert.zomaya@sydney.edu.au}

\begin{abstract}
Blockchain databases have attracted widespread attention but suffer from poor scalability due to underlying non-scalable blockchains. 
While blockchain sharding is necessary for a scalable blockchain database, it poses a new challenge named \emph{on-chain cross-shard database services}. 
Each cross-shard database service (e.g., cross-shard queries or inter-shard load balancing) involves massive cross-shard data exchanges, while the existing cross-shard mechanisms need to process each cross-shard data exchange via the consensus of all nodes in the related shards (i.e., on-chain) to resist a Byzantine environment of blockchain, which eliminates sharding benefits.

To tackle the challenge, this paper presents \textsc{GriDB}, the first scalable blockchain database, by designing a novel \emph{off-chain cross-shard mechanism} for efficient cross-shard database services.
Borrowing the idea of off-chain payments, \textsc{GriDB} delegates massive cross-shard data exchange to a few nodes, each of which is randomly picked from a different shard. 
Considering the Byzantine environment, the untrusted delegates cooperate to generate succinct proof for cross-shard data exchanges, while the consensus is only responsible for the low-cost proof verification.
However, different from payments, the database services' verification has more requirements (e.g., completeness, correctness, freshness, and availability); thus, we introduce several new \emph{authenticated data structures} (ADS).
Particularly, we utilize consensus to extend the threat model and reduce the complexity of traditional accumulator-based ADS for verifiable cross-shard queries with a rich set of relational operators.
Moreover, we study the necessity of inter-shard load balancing for a scalable blockchain database and design an off-chain and live approach for both efficiency and availability during balancing.
An evaluation of our prototype shows the performance of \textsc{GriDB} in terms of scalability in workloads with queries and updates.
\end{abstract}

\maketitle

\pagestyle{\vldbpagestyle}
\begingroup\small\noindent\raggedright\textbf{PVLDB Reference Format:}\\
\vldbauthors. \vldbtitle. PVLDB, \vldbvolume(\vldbissue): \vldbpages, \vldbyear.\\
\href{https://doi.org/\vldbdoi}{doi:\vldbdoi}
\endgroup
\begingroup
\renewcommand\thefootnote{}\footnote{\noindent
This work is licensed under the Creative Commons BY-NC-ND 4.0 International License. Visit \url{https://creativecommons.org/licenses/by-nc-nd/4.0/} to view a copy of this license. For any use beyond those covered by this license, obtain permission by emailing \href{mailto:info@vldb.org}{info@vldb.org}. Copyright is held by the owner/author(s). Publication rights licensed to the VLDB Endowment. \\
\raggedright Proceedings of the VLDB Endowment, Vol. \vldbvolume, No. \vldbissue\ %
ISSN 2150-8097. \\
\href{https://doi.org/\vldbdoi}{doi:\vldbdoi} \\
}\addtocounter{footnote}{-1}\endgroup

\ifdefempty{\vldbavailabilityurl}{}{
\vspace{.3cm}
\begingroup\small\noindent\raggedright\textbf{PVLDB Artifact Availability:}\\
The source code, data, and/or other artifacts have been made available at \url{\vldbavailabilityurl}.
\endgroup
}

\section{Introduction}
\label{sec:introduction}

% Depending on high security, transparent and traceability, blockchain technology has attracted a significant amount of attention in various areas, such as cryptocurrency~\cite{bitcoin}, supply chain~\cite{supply_chain}, international trade~\cite{international_trade}, etc. 
% Although blockchain is essentially a distributed database, it is far less convenient to use compared with traditional databases, such as MySQL, in terms of query types and efficiency. If a company wants to migrate its business from a traditional database to a blockchain, it needs to completely re-design its query interface and business logic, which is massively expensive and may sacrifice some of its functionalities to adapt to the blockchain.

Characterized by trustworthiness, transparency, and traceability, blockchain technologies have been integrated into many areas, such as cryptocurrency~\cite{bitcoin}, supply chain~\cite{supply_chain}, international trade~\cite{international_trade}, etc. 
In database management, blockchain technologies have attracted considerable interest in upgrading traditional databases to blockchain-empowered distributed databases~\cite{bc_vs_dd}, which forms an emerging research direction namely \emph{blockchain databases}.

Compared with traditional distributed databases, blockchain databases transact and record data via blockchains and construct an abstract database layer supporting various query functionalities on top of blockchains, which endow the distributed databases with immutability and traceability~\cite{blockchaindb, semantic, SEBDB, falcondb, vchain, gemtree,10.14778/3510397.3510406}. For example, BlockchainDB provides shared tables as easy-to-use abstractions as well as a key/value interface to read/write data stored in the blockchain~\cite{blockchaindb}. Pei \textit{et al.} introduces a Merkle Semantic Trie-based index to support semantic query, range query and fuzzy query on the blockchain~\cite{semantic}. SEBDB adds relational data semantics into blockchain storage and thus supports SQL query~\cite{SEBDB} and FalconDB presents a blockchain database with SQL query with time window~\cite{falcondb}.

Due to the underlying non-scale-out blockchains, most existing blockchain databases suffer from poor scalability.
For example, schemes in~\cite{SEBDB,falcondb} adopt Tendermint which achieves throughput of about 1000 transactions per second (TPS) but its network scale is less than 100. Schemes in~\cite{semantic,gemtree} adopt Ethereum aiming to support thousands of participants but only have tens of TPS.
% In comparison, the modern distributed databases such as CockroachDB~\cite{cockroachdb} and TiDB~\cite{tidb} can support thousands of TPS with more than 1000 nodes. 
The poor scalability makes the blockchain databases hardly meet the quality of service required in large-scale business in practice. 
% Thus, the shift towards a scalable blockchain database is inevitable.

% \begin{figure}[t]
% 	\centering
% 	\subfloat[][Sharding database]{
% 		\begin{minipage}[t]{0.43\linewidth}
% 			\centering
% 			\includegraphics[width=\linewidth]{introduction_1.pdf}
% 		\end{minipage}%
% 		\label{fig:introduction_1}
% 	}
% 	\hspace{+1pt}
% 	\subfloat[][Sharding blockchain database]{
% 		\begin{minipage}[t]{0.49\linewidth}
% 			\centering
% 			\includegraphics[width=\linewidth]{introduction_2.pdf}
% 		\end{minipage}
% 		\label{fig:introduction_2}
% 	}
% 	\caption{Illustration for different sharding systems.}
% \vspace{-14pt}
% \end{figure}

% Blockchain sharding is an architecture widely adopted by distributed databases for scalability~\cite{azure_sql,tidb,cockroachdb, mongodb}. 
Sharding is one of the most promising technologies for the blockchain scalability~\cite{elastico, omniledger, rapidchain, monoxide, sigmod_sharding, pldi, byshard}.
It divides the nodes into small groups called \emph{shards}, which can handle transactions in parallel and alleviate the storage overhead for each node. In such an approach, the transaction throughput scales linearly with the number of nodes. 
To develop a scalable blockchain database, this paper is going to construct an abstract database layer on a sharding blockchain by distributing database data and the corresponding task of storing, querying, or updating to different blockchain shards.
However, such a sharding for database storage and workload introduces a new requirement namely \emph{cross-shard database services}, i.e., \emph{data aggregation} for query and \emph{workload balancing} for management.
% As shown in \autoref{fig:introduction_2}, we define two roles named \emph{coordinator} and \emph{balancer} in the database layer responsible for these two designs, respectively.
%As shown in \autoref{fig:introduction_1}, there is a coordinator connecting to clients and aggregating the data in different shards and a balancer migrating the data among shards to achieve balanced workload. 
%Recently, sharding is also migrated to the blockchains for scalability~\cite{sok_sharding}, and some sharding blockchain prototypes can support thousands of TPS with thousands of blockchain nodes~\cite{elastico, omniledger,rapidchain,monoxide,pyramid,smart_contract_sharding}.
% Motivated by it, for a scalable blockchain database, we are going to construct an abstract database layer on top of a sharding blockchain in which nodes are divides into multiple shards and  each shard can handle transactions in parallel.
% However, 
% the transformation brings the new challenges of data and workload gap to the blockchain database.
% We summarize the challenges brought by the sharding architecture to the blockchain database as follows. 
% Although the similar challenges exist in traditional sharding databases as well, 
% the characteristics of blockchain make the existing schemes of sharding databases (i.e., coordinator and balancer) unsuitable for the blockchain databases due to the following challenges.

\begin{figure}[t]
	\centering
% 	\hfill
	\subfloat[][]{
		\begin{minipage}[t]{0.65\linewidth}
			\centering
            \includegraphics[width=\linewidth]{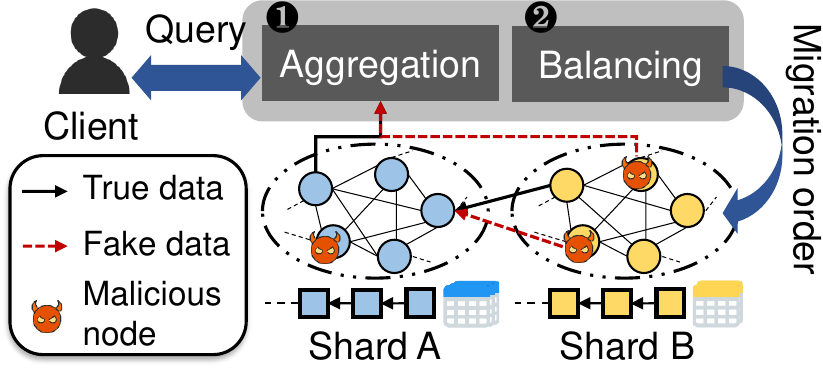}
    \end{minipage}%
    \label{fig:introduction_2}
	}
	\subfloat[][]{
		\begin{minipage}[t]{0.32\linewidth}
            \centering
			\includegraphics[width=\linewidth]{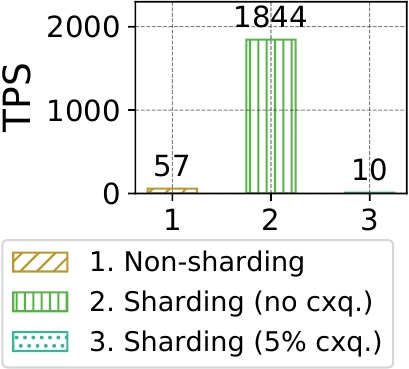}
		\end{minipage}
		\label{fig:preliminary}
	}
	\centering
    \caption{(a) Illustration for sharding blockchain database, which requires two new functions, i.e., data aggregation for query and workload balancing for management. 
    (b) Transaction throughput of non-sharding and sharding blockchain databases. (cxq. represents cross-shard query.)}
\end{figure}

As shown in \autoref{fig:introduction_2}, the data aggregation is caused by the sharding for storage. Particularly, the data related to a query may be stored by the blockchain nodes from multiple shards; thus, the query requires the involvement of several shards.
For example, if there are two tables stored in Shard A and B, respectively, then a SQL \texttt{JOIN} query combining these two tables involves both shards.
Moreover, the workload balancing is caused by the sharding for workload. 
Particularly, due to the sharding, each shard is only responsible for the workload of query and update to its storage. The demand imbalanced and dynamic nature of applications results in workload imbalance among shards, which significantly degrades the performance of sharding blockchain database.

The Byzantine environment of blockchain databases makes the technologies of traditional distributed databases no longer applicable. 
Malicious nodes may collude with each other and violate the protocol in arbitrary manners.
For example, in distributed databases~\cite{azure_sql}, a query can be easily realized by requesting from one database node in every related shard. 
However, in a sharding blockchain database, the correctness, completeness, and freshness of cross-shard queries can hardly be guaranteed when accidentally requesting from a malicious node. 
Moreover, elastic workload balancing is the first-class feature in modern databases~\cite{estore} achieved by load migration. 
However, unlike one-to-one crash-tolerant migration in most distributed databases, the sharding blockchain database requires a many-to-many migration across shards in which malicious nodes can intercept, tamper or forge the migrating table.

To resist the Byzantine faults for cross-shard database services, an intuitive idea is to process them through the cross-shard mechanism of blockchain sharding. 
% However, the existing technologies for cross-shard transactions in the blockchain sharding (called \emph{cross-shard mechanisms}) do not work due to high expenses. 
In detail, the core of cross-shard database services is to transfer tables among shards despite Byzantine failures. 
The cross-shard mechanism guarantees that each data transfer (e.g., money transfer in the conventional blockchain) is agreed by the majority of honest nodes in all its related shards.
Transferring a table among shards through the mechanism
% invoke the cross-shard mechanism of the existing blockchain sharding. 
% Depending on cross-shard mechanism in the traditional blockchain sharding, a transaction can be committed in multiple shards despite Byzantine failures.
% Thus, for a query or a migration involving two shards, a shard can transfer its data to another shard in the form of transactions via cross-shard mechanism which 
can guarantee that the table transfer will not be compromised (detailed in \autoref{sec:preliminary_sharding} and \autoref{sec:strawman}).
However, such an idea is costly.
On the one hand, each query or migration involves a massive set of semantics-related data (e.g., rows belonging to a table) in a blockchain database.
On the other hand, all existing cross-shard mechanisms are on-chain (i.e., requiring the consensus of all the related shards).
Therefore, all nodes in the related shards need to participate in the consensus on numerous data.
As proved in \autoref{fig:preliminary}, the existing on-chain cross-shard mechanisms cannot support even 5\% cross-shard queries in a sharding blockchain database with 32 shards (detailed in \autoref{sec:experiment}).
% rather than a simple account balance in the traditional blockchain.
% each mechanism call is expensive since it needs to split each transaction into several sub-transactions to commit, which multiplies the overhead. 
% the transactions in the blockchain database have richer semantics information and functionality, such as relational semantics data structure and the corresponding query. 
% Therefore, such an on-chain solution is expensive.
% the sharding blockchain database needs a more efficient design to bridge the data and workload gap.

To this end, this paper focuses on relational blockchain database and proposes the first relational sharding blockchain database, named \textsc{GriDB}. 
In comparison with the previous blockchain databases, \textsc{GriDB} guarantees high scalability while providing support for relational database services in blockchain sharding.
% We emphasize that it is a \emph{layer 2} solution which means it can be constructed on top of the existing sharding blockchains. 
Motivated by the idea of off-chain payments and verifiable computing, \textsc{GriDB} enables an off-chain execution of the cross-shard database services by adopting \emph{authentication data structure} (ADS) to delegate cross-shard communication-intensive tasks to a few nodes in a verifiable manner.
% It aims to solve the challenges of the query processing and data management in a sharding blockchain database. 
% First, we introduce relational data semantics and query functionality into blockchain transactions to abstract a sharding blockchain as a distributed relational database.
% To achieve inter-shard load balancing, we design an off-chain live migration protocol that migrates workload among shards with little service disruption and migration overhead.
% Moreover, some modules related to query optimization and load balancing scheduler are presented to complete the solution of a sharding blockchain database. 
We summarize our contributions as follows.
\begin{itemize}
    \item \textsc{GriDB} introduces relational data semantics and query functionality into blockchain transactions to abstract a sharding blockchain as a distributed relational database. The clients can send requests to any untrusted blockchain nodes for storing, manipulating and retrieving data.
    \item To provide a query layer of abstraction on sharded data, we design a cooperative delegation-based approach with a constant complexity of data transfer among shards without sacrificing security. It delegates the tasks of data aggregation to a few nodes in different shards
    % by verifiable set operations 
    and constructs a succinct proof used to on-chain verify. 
    \item To meet the dynamically skewed workloads and achieve inter-shard balancing, we propose an off-chain live migration that migrates the database service among shards with security, low cost, and minimum interruption. 
    % It includes a cross-shard notification mechanism for synchronization during migration.
    \item We develop a prototype for \textsc{GriDB} and conduct a comprehensive evaluation. The result shows that \textsc{GriDB} achieves a scalable throughput for SQL linearly increasing with the shard number compared with the non-sharding works.
\end{itemize}

% The remainder of this paper is structured as follows. \autoref{sec:preliminary} provides some background. \autoref{sec:system_data_model} describes the system overview. \autoref{sec:system_design} presents the system design and theoretical analysis. \autoref{sec:design_refinement} gives several design refinements. 
% \autoref{sec:experiment} reports evaluation results. 
% \autoref{sec:related_work} provides related work. 
% \autoref{sec:conclusion} concludes the paper.

\section{Preliminaries}
\label{sec:preliminary}

\subsection{Blockchain Sharding}
\label{sec:preliminary_sharding}

Blockchain sharding achieves scalability by dividing the blockchain nodes into multiple shards, each of which is responsible for receiving, validating, and processing part of transactions. 
A blockchain sharding scheme is generally composed of \emph{shard formation}, \emph{intra-shard consensus} and \emph{cross-shard mechanism} as follows. 

1) At the beginning, each node is assigned to a shard. For example, in Elastico~\cite{elastico}, a node generates an identity by solving a Proof-of-Work (PoW) puzzle to avoid Sybil attack and is assigned to a shard with an ID equal the last several bits of the identity. 

2) In each round, the nodes within a shard run an intra-shard consensus to agree on the same block consisting of valid transactions. 
For example, some systems~\cite{elastico, omniledger, rapidchain} adopt a Byzantine fault tolerant (BFT) protocol (e.g., PBFT~\cite{pbft} and collective signing BFT~\cite{cosi}) as intra-shard consensus to resist malicious nodes. 
An intra-shard consensus should satisfy \emph{safety}, i.e., the honest nodes agree on the same valid block in each round, and \emph{liveness}, i.e., the block in each round will eventually be committed or aborted.

3) Due to the sharding, some transactions may involve the state of more than one shard and thus are called \emph{cross-shard} transactions. 
The core idea of most of the existing cross-shard mechanisms is to divide each cross-shard transaction into several sub-transactions, each of which is processed by a shard. 
% For example, ...
Then, the shards handle them with the guarantee of ACID, i.e., atomicity, consistency, isolation, and durability, for every cross-shard transaction. 
In other words, if a cross-shard transaction is committed by the cross-shard mechanism successfully, it means that the majority of honest nodes in every related shard receive the same transaction and agree that it is valid. 
% }
% In other words, the cross-shard mechanism guarantees that each valid cross-shard transaction can be committed to all the related shards.

% $\mathrm{cxMechanism}(S, txn)$ ... cross-shard transaction $txn$ related to a set of shards $S$ 

% \begin{algorithm}[t]
%     \SetKwInput{KwIn}{Input}
%     \SetKwInput{KwOut}{Output}
%     \SetKwFunction{FMain}{Cross-shard Mechanism}
%     \SetKwProg{Fn}{Function}{:}{}
% 	\caption{Cross-Shard Transaction Processing}}
%     \label{alg:cross_shard_mechanism}
%     \Fn{\FMain{$txn$}}{
%         \KwIn{a transaction $txn$ related to a set of shards $S$}
%         \If{$s$ receives the same $txn$ \textbf{and} $txn$ passes the validation of any $s \in S$}{
%             $txn$ can be committed in all $s \in S$
%         }
%         \Else{
%             $txn$ will not be committed in any $s \in S$
%         }
%     }
% \end{algorithm}

% \vspace{-2pt}

\subsection{Authenticated data structures}
\label{sec:preliminary_vso}

%For the ADS generation, we first refine the system model and data model of \textsc{GriDB} as follows. First, each block has a control transaction named \emph{version transaction} to denote the version of tables. It includes the latest Merkle tree roots of all tables in the shard. Second, each node is a light node for the other shards, which means it will store the block headers of the other shards. Note that these two refinements introduce not much storage cost, which is evaluated in the experiment.

\textbf{Verifiable set operation (VSO).} VSO~\cite{verifiable_set_1,verifiable_set_2} is an ADS which enables clients to outsource set computation tasks (e.g., intersection and union) to an untrusted server.
%Fix two cycle multiplicative groups $\mathbb{G}$ and $\mathbb{G}_T$ of prime order $p$, $\mathbb{G}$'s generator $g$, and a bilinear map $e:\mathbb{G} \times \mathbb{G} \rightarrow \mathbb{G}_T$. With the tuple of bilinear pairing parameters $(p, \mathbb{G}, \mathbb{G}_T,e, g)$,
%to provide a proof of membership for elements that belong to a set. 
The server returns the result with an ADS, depending on which, one can verify the result without downloading the original data and re-calculating by itself. 
VSO consists of the following probabilistic polynomial-time algorithms:
%We introduce the workflow of the method as follows.

% Based on a bilinear-map accumulator primitive \cite{accumulator_pri}, the method begins by choosing randomly a value $s \in \mathbb{Z}_p$ and letting $sk=s$ be the secret key and $pk=(g^s, \cdots, g^{s^q})$ be the public key, where $g$ is the generator of a cycle multiplicative group $\mathbb{G}$ and $q$ is an upper-bound on the cardinality of sets in the algorithm. 
% For a set $X \subset \mathbb{Z}_p$, the algorithm can compute its accumulation value by $acc(X)=g^{\prod \limits_{x \in X}(x + s)}$. 
% Given two sets $X^i$ and $X^j$, $(X^{*}, \pi) \leftarrow \mathrm{prove}(X^i, X^j, pk)$} returns the intersection (or union) result of these two sets $X^*$ with a proof $\pi$. 
% The client can validate the query result using $\{\mathrm{accept}, \mathrm{reject}\} \leftarrow \mathrm{verify}(acc(X^i), acc(X^j), X^*, \pi, pk)$}. Note that the client only needs two accumulation values $acc(X^i)$ and $acc(X^j)$ and a proof to verify the result instead of downloading the sets $X^i$ and $X^j$.

\begin{itemize}
    \item $(sk, pk) \leftarrow \mathrm{genKey}(1^\lambda)$: Let $\lambda$ denote the security parameter. Based on a bilinear-map accumulator primitive \cite{accumulator_pri}, it begins by choosing randomly a value $s \in \mathbb{Z}_p$ and letting $sk=s$ be the secret key and $pk=(g^s, \cdots, g^{s^q})$ be the public key, where $g$ is the generator of a cycle multiplicative group $\mathbb{G}$ and $q$ is an upper-bound on the cardinality of sets.
    \item $acc(X) \leftarrow \mathrm{setup}(X)$: For a set $X \subset \mathbb{Z}_p$, it computes the accumulation value of $X$ by $acc(X)=g^{\prod \limits_{x \in X}(x + s)}$.
    \item $(X^{*}, \pi) \leftarrow \mathrm{prove}(X^i, X^j, pk)$: Given two sets $X^i$ and $X^j$, it returns the intersection (or union) result $X^*$ and a proof $\pi$.
    \item $\{\mathrm{accept}, \mathrm{reject}\} \leftarrow \mathrm{verify}(acc(X^i), acc(X^j), X^*, \pi, pk)$: On input the accumulation values $acc(X^i)$, $acc(X^j)$, a proof $\pi$, and the public key $pk$, it return $accept$ if and only if $X^i \cup X^j  = X^*$ for intersection (or $X^i \cap X^j  = X^*$ for union).
\end{itemize}

% The client only needs the accumulation values and a proof to verify the result instead of downloading the whole sets.

The unforgeability for VSO has been proved to be held under the q-SBDH assumption~\cite{qsbdh} in
bilinear groups. For more details about the proof, refer to~\cite{verifiable_set_1, verifiable_set_2}. Additionally, the case of the intersection (or union) for an arbitrary number of sets is similar.

\textbf{Merkle tree.} Merkle tree~\cite{merkle1987digital} is an ADS which enables clients to verify the correctness of committed transactions. Its leaf nodes are composed of hash values of transactions, and non-leaf nodes are generated upward through hash operations until the root named \emph{Merkle root} is generated. 
% Thus, it can summarize all transactions in a block and generate a digital fingerprint for them, i.e., Merkle root. 
Each block header stores a Merkle root for its contained transactions; thus, the clients only need to synchronize the headers to verify all committed transactions. 
A blockchain node can provide the clients with any transaction and a Merkle proof consisting of the siblings path of the transaction in the tree. 
The clients can validate the transaction by constructing a Merkle root based on the proof and comparing it with the locally stored one.

\section{System Model}
\label{sec:system_data_model}

\textbf{System Components.} In \textsc{GriDB}, there are two types of entities:

1) The \emph{database clients} are the service consumers of \textsc{GriDB}. They neither participate in the consensus nor store the whole content of blockchains locally because they are often lightweight devices.
% such as Internet-of-Things (IoT) devices.
% Furthermore, we assume that each client can read or write any parts of the database in this section, and postpone the design of access control in \textsc{GriDB} to \autoref{sec:access_control}.

2) The \emph{blockchain nodes} are the consensus participants for the blockchain and are divided into a number of shards. Each node is responsible for verifying, processing and storing transactions of its located shard. 
\textsc{GriDB} is a layer-2 database framework constructed on top of existing blockchain sharding systems, and it is \emph{sharding-agnostic}, which means the underlying system can adopt any sharding schemes (including shard formation, intra-shard consensus and cross-shard mechanism) from~\cite{omniledger,rapidchain,monoxide,pyramid}. 
However, the underlying blockchain sharding system's intra-shard consensus should satisfy both safety and liveness, and its cross-shard mechanism guarantees the ACID of each cross-shard transaction (see \autoref{sec:preliminary_sharding}).
We emphasize that the cross-shard mechanism of the underlying blockchain sharding system is one of the important components of the off-chain cross-shard mechanism of \textsc{GriDB}, which will be described in the following sections.
To avoid confusion, we will call the former \emph{on-chain cross-shard mechanism}.
% The intra-shard consensus guarantees that the nodes in a shard can agree on a sequence of valid transactions and the cross-shard mechanism guarantees that each valid cross-shard transaction can be committed to all the related shards.

\textsc{GriDB} considers an outsourced database scenario~\cite{survey_outsourced, survey_outsourced_2} in which the clients outsource their data management to the blockchain. The nodes host the client's databases and the clients send requests to the nodes to create, store, update and query their databases.

\textbf{Threat Model.} The threat model of \textsc{GriDB} relies on that of the underlying blockchain sharding composed of two kinds of blockchain nodes: \emph{honest} and \emph{malicious}. The honest nodes abide by all protocols in \textsc{GriDB} while malicious nodes may collude with each other and violate the protocols in arbitrary manners, such as denial of service, or tampering, forgery and interception of messages. Although there are malicious nodes in the shards, the sharding blockchains~\cite{rapidchain, omniledger, monoxide} can guarantee that each shard is trusted with high probability, i.e., the result published by any shards is trusted. 
Different from~\cite{blockchaindb}, \textsc{GriDB} does not require a strong assumption that each client trusts the nodes it connect.

\textbf{Transaction Model.} The requests of clients are processed in the form of blockchain transactions that are divided into two types: \emph{data} and \emph{query transaction}. The first one is used to update (such as insert, update, and delete) the database state and the second one is used to query the database state.
Besides, in \autoref{sec:live_migration}, there are some \emph{control transactions} used to support the database management such as database migration. 
The type division will not affect the compatibility for the underlying blockchain because there is a ``data'' field in the transactions of most blockchains, and \textsc{GriDB} places different data in the field for different types of transactions. 
The details of transactions will be described as follows. 

\section{GriDB Overview}

\subsection{System Overview}
\label{sec:system_overview}

\begin{figure}[t]
    \centering
    \includegraphics[width=\linewidth]{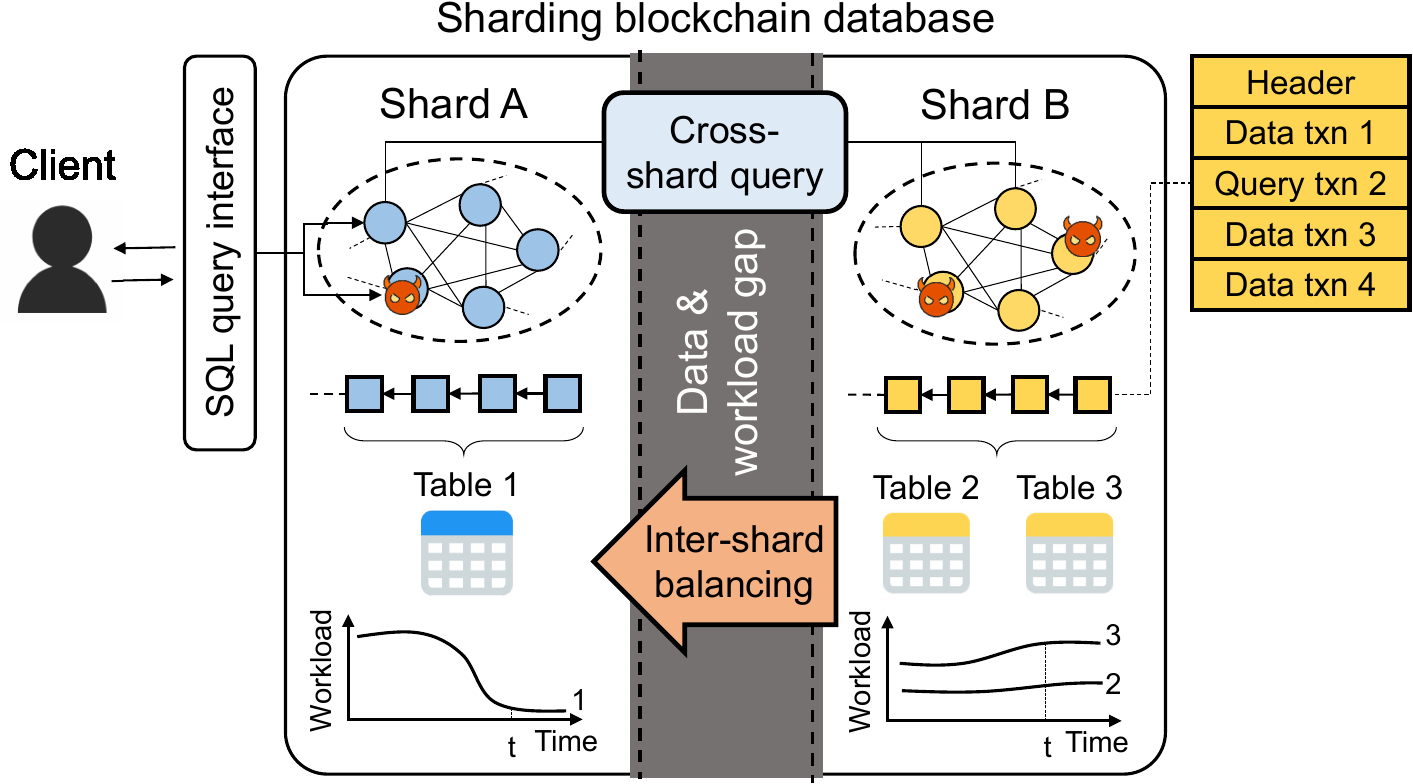}
    \caption{System overview for \textsc{GriDB}.}
    \label{fig:framework}
% \vspace{-16pt}
\end{figure}

As shown in \autoref{fig:framework},  \textsc{GriDB} considers a distributed relational database storing a number of tables. 
For scalability, the workload of each table is distributed to a shard (A fine-grained sharding for blockchain database through table partition will be described in \autoref{sec:fine_grained_sharding}.)
If a client decides to query or update a table, it can issue requests to any nodes in the shard responsible for the table. 
% ...}
As described in the threat model, the nodes receiving the requests may be malicious, thus they are required to return a proof used for authentication.
To generate a proof, based on the request received, a node first proposes a transaction which can be one of the following.

1) A data transaction is used to manipulate (such as insert, update, and delete) the data and includes an \texttt{INSERT}, \texttt{DELETE} or \texttt{UPDATE} SQL statement. In \textsc{GriDB}, the data is in the form of a relational model. The same type of data has a unified semantic description as a schema composed of several attributes. 
% Each transaction has a table id to denote which table it belongs to.
% The data transactions with the same table id and in the same shard belong to the same table.
An insert statement inserts new data based on explicitly specified values or from the existing data via a nested subquery.
Since blockchain is append-only, a delete statement is implemented by marking old data as invalid, i.e., cannot be queried, and an update one is implemented by a sequence of delete and insert operations to overwrite. 
Based on the data transactions recorded in the blockchain, each node in a shard maintains a tamper-proof copy of a relational database.

2) A query transaction is used to record query results to clients.
A query statement begins with a keyword \texttt{SELECT} followed by a subset of column names, and then a keyword \texttt{FROM} followed by a table (or a \texttt{JOIN} sub-clause used to combine multiple tables in a later section.) Following these, a \texttt{WHERE} clause is followed by a sequence of predicates connected by logical operators (e.g., \texttt{AND}, \texttt{OR}, \texttt{NOT}) that restrict the rows used when computing the output. After processing the request, the node can put the query statement and result into the data field of a query transaction. 
To avoid occupying too much on-chain storage, the query results can be offloaded to an off-chain storage and the query transactions only store the hash of query results, according to which the clients can download the correct results from the off-chain storage based one the hash.

Next, the node broadcasts the proposed transaction to the network. If a data transaction is committed to the blockchain, the majority of honest nodes accept and execute the data transaction, which means the database state has been successfully updated. 
If a query transaction is committed, the majority of honest nodes agree on the query results. Finally, the client can authenticate the result returned by its connected node by validating if the transaction for its request is committed. This transaction validation for clients has been implemented in most blockchains, such as Simplified Payment Verification (SPV) in Bitcoin and Ethereum. 

In addition to a Merkle tree storing transactions like traditional blockchains, to enable every node or client to know every table's location, every node maintains an additional Merkle tree. 
% The one includes blockchain transactions similar to traditional blockchain.
The tree stores the location of all tables in the form of \emph{table name-shard id} pairs. 
It is updated in each epoch according to a global cross-shard control transaction including a new planning strategy for the following epoch (refer to \autoref{sec:scheduler}).
Thus, depending on the tree, every node or client can find the correct shard storing the target table.

\subsection{Challenges} 
\label{sec:challenges}

Dividing the tables across different shards improves the blockchain database's scalability. 
However, it is not enough for a scalable blockchain database due to two following problems.

\textbf{Problem 1 (cross-shard query):} A client's query involving tables only in a single shard can be served by one shard because each node of the shard can validate and agree on the query result in an intra-shard consensus.
However, a request to query the tables from different shards cannot be completed by a single shard and requires the cooperation of multiple shards. For example, in \autoref{fig:framework}, a query joining on Table 1 and 2 involves the data of Shard A and B thus cannot be completed by only one of them.
%Thus, as shown in Component (a) in \autoref{fig:framework}, we design an efficient cross-shard query service based on a consensus-empowered ADS in \autoref{sec:system_design}.
% Besides, the multi-table deletion/update and the insertion/deletion/update with subqueries may result in cross-shard data transactions.
% \footnote{\textsc{GriDB} considers single-table insert, delete and update operations.}, 
% insert, delete or update SQL statement modifies the state of a single table. 
% explicitly specifies values or their nested subqueries involving tables in the same shard, they can be processed by a shard. The case of 
% if it denotes an insert, delete or update statement with nested subqueries or a multiple-table delete/update statement~\cite{mysql_ref}, it may involves multiple shards.

\textbf{Problem 2 (inter-shard workload imbalance):} 
% Similar to the traditional distributed databases, 
It is hard to guarantee that the workload of every table is the same and static, thus some shards can be overloaded while some others remain idle. For example, in \autoref{fig:framework}, Shard A is responsible for Table 1 and Shard B is responsible for Table 2 and 3. At the beginning, the total workload in Shard A and B is similar. 
However, if the workload of Table 1 drops over time, Shard A becomes idle. 
To fully utilize the throughput, dynamically migrating the workload among shards and alleviating the effect of hotspots are crucial.
% In particular, if Table 2 can be migrated from Shard B to A, the workload of both shards can be balanced. 
% Besides, ...
% However, the traditional load migration in distributed database~\cite{zephyr,squall} targets only the crash fault tolerance (CFT) model
% in which there are no malicious nodes 
% thus cannot fit the Byzantine environment in blockchain. 

\section{System Design}
\label{sec:system_design}

% In this section, 
% Next, to address the strawman system's drawbacks, we discuss the causes and introduce two key designs of our off-chain cross-shard mechanism in \autoref{sec:cross_shard_query} and \autoref{sec:live_migration}.

\subsection{Strawman}
\label{sec:strawman}

For the two challenges above, we first describe a strawman sharding blockchain database only based on the on-chain cross-shard mechanism of the existing sharding systems as follows. 

% \begin{figure}[t]
% 	\centering
% 	\subfloat[][Shard-cooperation cross-shard query]{
% 	\begin{minipage}[t]{0.83\linewidth}
% 			\centering
% 			\includegraphics[width=\linewidth]{strawman_cross_shard_query.pdf}
% 		\end{minipage}
% 	\label{fig:strawman_cross_shard_query_2}
% 	}
	
% 	\vspace{+10pt}
% 	\subfloat[][Stop-restart inter-shard migration]{
% 	    \begin{minipage}[t]{0.8\linewidth}
% 			\centering
% 			\includegraphics[width=\linewidth]{strawman_migration.pdf}
% 		\end{minipage}%
% 	    \label{fig:strawman_cross_shard_migration}
% 	    }
% 	\caption{Overview for the strawman.}
% % \vspace{-16pt}
% \end{figure}

Consider a cross-shard query involving two tables, i.e., Table 1 and 2, located in Shard A and B, respectively. 
For such a cross-shard query, we present a \emph{shard-cooperation approach} based on the on-chain cross-shard mechanism. 
% As shown in \autoref{fig:strawman_cross_shard_query_2}, 
Shard A first commits many cross-shard data transactions involving Shard A and B, including the data of Table 1 via the cross-shard mechanism. 
As described in \autoref{sec:preliminary_sharding}, the mechanism guarantees the transactions can be committed in Shard A and B. 
Then, Shard B can get Table 1 from the transactions, compute the query result, and commit a query transaction with the query result.
% However, considering that blockchain is Byzantine, for the security, each node in Shard B needs to download Table 1 to validate the query transaction. 
However, when there are many cross-shard queries, the table transfer among shards will be frequent, resulting in system overloaded and network blocked.

For the inter-shard load balancing, the latest version of the table should be transferred from the source shard to the destination shard. 
If there is data being left out, the completeness of queries on the table cannot be guaranteed after migration. 
Thus, we present a \emph{stop-restart migration approach} based on the on-chain cross-shard mechanism as follows.
% As shown in \autoref{fig:strawman_cross_shard_migration}, 
Shard A first stops processing new transactions for the table via an intra-shard consensus, which avoids the migrating table being modified during migration.
Then, Shard A commits many cross-shard data transactions to reconstruct the latest version of the table in Shard B. After all transactions are committed, Shard A commits a cross-shard transaction to mark the end of the migration, and Shard B can restart the service of the table.
However, when the migrating table is enormous, the approach incurs a high penalty due to the prolonged service interruption for the migrating table and the influence on other tables' throughput.

To solve these drawbacks, we introduce two key designs for our off-chain cross-shard mechanism in \autoref{sec:cross_shard_query} and \autoref{sec:live_migration}.

\subsection{Cross-Shard Query Authentication}
\label{sec:cross_shard_query}

% \begin{figure}[t]
%     \centering
%     \includegraphics[width=3.1in]{cross_shard_query.pdf}
%     \caption{Illustration for delegation-based cross-shard query.}
%     \label{fig:cross_shard_query}
% % \vspace{-16pt}
% \end{figure}

\textbf{Motivation.} Although the above shard-cooperation approach is safe, it is expensive since the table transfer among shards for each cross-shard query is between groups of nodes (to guarantee that there is a majority of honest nodes for each shard participating). 
An intuitive idea to avoid this overhead is to pick one delegate from each shard. 
Then, for each cross-shard query, a delegate downloads the related tables from the other delegates and evaluates the query results. 
However, any malicious delegate can easily tamper with the query result by providing a fake or out-of-date table.
%The other nodes in a shard cannot verify the cross-shard query result because they do not have the related tables in other shards.
%, thus the security of the query cannot be guaranteed. 

We aim to design an ADS to allow the delegates to prove the validity of cross-shard query results.
% To implement the idea of delegation securely and efficiently, 
% we propose a \emph{node-cooperation approach} in which 
% In an outsourced database~\cite{survey_outsourced, survey_outsourced_2},
% ADS helps an untrusted server prove to clients that its query results are valid.
%for a succinct proof of cross-shard query result to enable the other nodes in the related shards to validate the results provided by the untrusted delegates.
The existing outsourced databases have designed some ADS for SQL~\cite{survey_outsourced,survey_outsourced_2}.
For example, for a \texttt{JOIN} query involving the same column of two tables, a node treats the columns of these two tables as two sets and constructs a proof for their intersection through VSO~\cite{integridb}. 
However, the existing ADS for SQL cannot be applied in our cross-shard query due to two difficulties.
First, different from the outsourced database in which there is no sharding and a server stores the whole data copy, any delegate in \textsc{GriDB} only stores the tables of its located shard and downloads tables from the other untrusted delegates and thus cannot construct a valid proof by itself.
% (i.e., different threat model). 
% Each cross-shard query requires the cooperation among multiple delegates. 
Second, to support arbitrary verifiable SQL queries, besides VSO, the other outsourced databases need to adopt interval trees~\cite{integridb} or zero-knowledge proof~\cite{vsql}, which costs tens of minutes for a query~\cite{vsql} due to high computation complexity. 
% These approaches may significantly increase the query latency of the blockchain database. 

Thus, we propose a \emph{delegation-based approach} by integrating VSO with the intra-shard consensus to implement an efficient and secure ADS for arbitrary SQL query in \textsc{GriDB}. 
Its main idea is to divide each query into some algebra operators with different input data.
Particularly, it validates the operators involving multiple shards' data through VSO and those involving single shard's data through the intra-shard consensus.
% It then adopts intra-shard consensus for the validity of operators 
The cross-shard query validity can finally be proved through a chain of trust, i.e., proving the validity of every operator from beginning to end.
% divide each query into \emph{cross-shard} and \emph{intra-shard operators}. 
% Since the validity of cross-shard operators requires massive cross-shard data exchanges, we delegate them to the delegates via VSO.
Such a manner makes the best use of the advantage of both the shard-cooperation approach (i.e., low computation) and the existing ADS for verifiable SQL (i.e., low communication) and bypasses their disadvantage.
% ...

\textbf{Design.} The overall cross-shard query procedure is given in \autoref{alg:cross_shard_query}. For each cross-shard query, we identify the related shard of the table following \texttt{FROM} as \emph{main shard} and the shards of the tables following \texttt{JOIN} as \emph{sub shards}.
A client can issue a cross-shard query request to any nodes in the main shard.
Next, one node is chosen from each related shard (Line 2), which can be round-robin or randomly by a verifiable random function~\cite{algorand}. 
The malicious or low-response delegates can be replaced by a view change similar to PBFT.
The delegated node in the main shard is called the \emph{main node} denoted by $\mathbb{M}$ and those in the sub shards are called \emph{sub nodes} denoted by $\mathbb{S}$.
The main node downloads each involved table for the query from the sub node in the corresponding sub shards (Line 3).
After downloading all involved tables, the main node evaluates the query result and generates a proof (Lines 4, 8-14).

To generate the proof, the main node first translates each SQL query into a relational algebra
tree composed of algebra operators~\cite{db_book}, e.g., the right part of \autoref{fig:query_evaluation}.
Each node in the tree denotes a unary (or binary) algebra operator taking one (or two) inputs, applying a function, and outputting its result to the next operator. The edges represent data flow from bottom to top.

\begin{algorithm}[t]
    \SetKwInput{KwIn}{Input}
    \SetKwInput{KwOut}{Output}
    \SetKwFunction{FProof}{genProof}
    \SetKwFunction{FCX}{validateCx}
    \SetKwProg{Fn}{Function}{:}{}
	\caption{Cross-Shard Query Authentication}
    \label{alg:cross_shard_query}
        \KwIn{query request $Q$ involving tables in a set of shards $S$}
        \KwOut{query result $R$, verification object $VO$}
        Delegates $\mathbb{M}$ and $\mathbb{S}$ are selected from $S$\\
        $\mathbb{M}$ downloads the related tables from $\mathbb{S}$\\
        $\mathbb{M}$ evaluates query result $R$ and get proof $\Upsilon$ via $\textsf{genProof}(Q)$\\
        $\mathbb{M}$ proposes a cross-shard query transaction $txn$ involving $S$ and including $R$ and $\Upsilon$\\
        \If{\FCX{$S, txn$} == True}
        {
            % $txn$ is committed in the blockchains of $S$\\
            $VO \leftarrow $ the list of SPV proofs in $S$ for $txn$
        }
        \Fn{\FProof{$Q$}}{
            \For{\textup{cross-shard operator} $op \in Q$}{
                Set $C^i$ and $C^j$ as the columns involved by $op$ and $pk$ as the public key\\
                $(C^{*}, \pi) \leftarrow \mathrm{prove}(C^i, C^j, pk)$\\
                Get $bm^i$ and $bm^j$ based on $C^i$, $C^j$ and $C^*$\\
                Add $\langle acc(C^i),acc(C^j), \pi, bm^i, bm^j \rangle$ to $\Upsilon$
            }
            return $\Upsilon$
        }
        \Fn{\FCX{$S, txn$}}{
            % \tcc{...}
            \For{\textup{shard} $s \in S$}{
                \If{$\Upsilon$ or $R$ \textup{is invalid}}{
                % \If{$\Upsilon$ \textup{is invalid} or }{
                    return $False$
                }
                $txn$ is committed in the blockchain of $s$
            }
            return $True$
         }
\end{algorithm}

In the tree, we identify join (or union) operators involving tables in different shards as \emph{cross-shard operators} and the others as \emph{intra-shard operators}. 
Each intra-shard operator can be processed by the nodes of the corresponding shard based on their stored tables. 
In comparison, each cross-shard operator involves the data gap among shards, thus the main node needs to generate a proof. 
The proof is composed of the accumulation values (refer to \autoref{sec:preliminary_vso}) of the corresponding columns in the tables to be joined or unioned, a VSO proof, and a position indicator for the intermediate result (or final result). 
The position indicator is a bitmap to indicate which rows are chosen in a table. For example, considering the SQL statement in \autoref{fig:query_evaluation}, 
% the main node first process the selection operator for the two tables. 
we denote the oid columns of these two tables after processing the selection operators as $C^i$ and $C^j$, respectively. 
The generated proof $\Upsilon$ is $\langle acc(C^i),acc(C^j), \pi, [1,0,1], [1,0,0]\rangle$. 
The position indicators $[1,0,1]$ and $[1,0,0]$ mean that the first and third rows in Table 1 and the first row in Table 2 are chosen, respectively. A cross-shard query may include multiple cross-shard operators thus the main node will produce a list of proofs  $\Upsilon$, each of which is for a cross-shard operator.

After the query result and the corresponding proof are generated, the main node proposes a cross-shard query transaction, including the result and the proofs and involving the related shards (Line 5). 
Then, to validate the cross-shard query, each related shard runs an intra-shard consensus on the transaction by evaluating each algebra operator for their stored tables and verifying the proofs related to the tables of the shard (Lines 15-20). 
% Note that for the operators except the cross-shard operators, the nodes of the corresponding shard can independently handle them.
For example, as shown in \autoref{fig:query_evaluation}-\ding{184}, during the consensus on the cross-shard query transaction, each node can validate and execute intra-shard operators based on the local data and validate and execute cross-shard operators based on the VSO proof.
During the validation, they can optimistically assume that the accumulation values related to the other shards' tables are valid.
Finally, if the cross-shard query transaction passes the validation of every related shard (Line 6), it will be committed in the blockchains of all related shards and the client can accept the query result included in the transaction via SPV (Line 7).
Besides, if a malicious main node sends different copies of a transaction to shards, the client can detect the inconsistency by checking the Merkle proofs of the transaction via SPV (Line 7).

\begin{figure}[t]
    \centering
    \includegraphics[width=\linewidth]{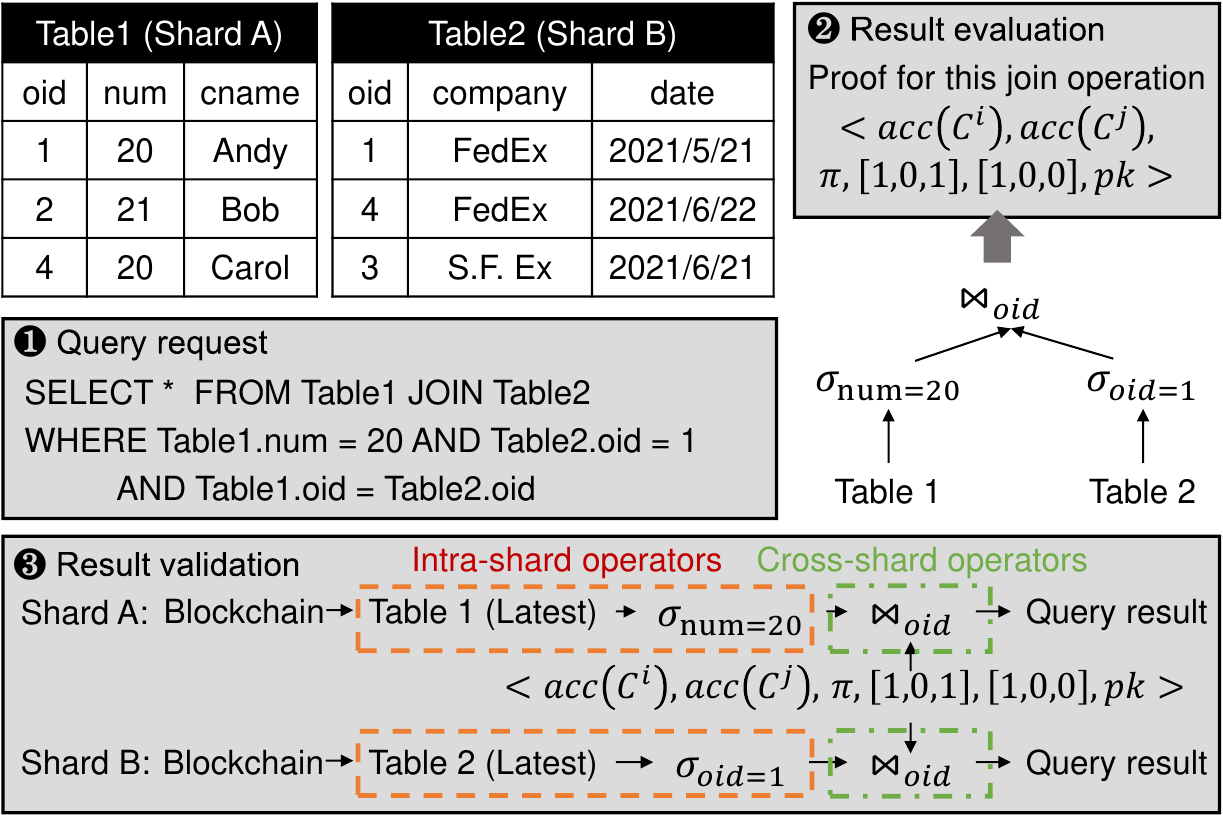}
    \caption{Example for ADS proof generation of two tables distributed in Shard A and B, respectively. ($\sigma$ is an operator to select rows from a relation and $\Join$ is an operator to join tables based on a specified column.)}
    \label{fig:query_evaluation}
% \vspace{-16pt}
\end{figure}

%\textbf{Example.} \texttt{JOIN} and \texttt{UNION} in SQL statement consider the columns in several tables as the collections of data and compute their intersection and union, respectively. 
% Consider the case where a client has a cross-shard query involving two columns of two tables stored in Shard A and B, respectively. We denote the columns of these two tables as $C^A$ and $C^B$, respectively. 
% The picked nodes can construct a cross-shard query transaction including $<acc(C^A),acc(C^B), hash(result)>$.
% The nodes in Shard A validate $acc(C^A)$, the query result, and the proof $\pi$ with the assumption that $acc(C^B)$ is valid. At the same time, the nodes in Shard B validate $acc(C^B)$ and the query result.
% Recall that the cross-shard transaction mechanism can guarantee that only if a cross-shard transactions is valid in all related shards, the cross-shard transaction can be committed.

\textbf{Security Analysis.} The analysis relies on the intra-shard consensus of blockchain sharding thus we define $v$ as the fault threshold~\cite{multi_threshold} of the adopted blockchain sharding in \textsc{GriDB}. 
For example, Rapidchain~\cite{rapidchain} tolerates up to $v=1/2$ Byzantine faults, while the asynchronous or Omniledger~\cite{omniledger} tolerates only up to $v=1/3$ Byzantine faults. 
Next, we describe the formal definition~\cite{vsql,integridb} of our cross-shard query's security as follows.
\begin{definition}
A query is secure if any polynomial-time adversary's success probability is negligible in the following experiment:

For a query $q$, the adversary is picked as main node or sub node for the generation of query transaction including result $R$. The adversary succeeds if the query transaction is committed in all related shards and one of the following results is true: 1) $R$ includes a row which does not satisfy $q$ (\textbf{correctness}); 2) There exist a row which is not in $R$ but satisfies $q$ (\textbf{completeness}); 3) R includes a row not from the latest tables generated by all the committed data transactions (\textbf{freshness}).
\end{definition}

\begin{theorem}
Our proposed cross-shard query mechanism satisfies the security property as defined in Definition 1 if the proportion of malicious nodes in each shard is no more than the fault threshold $v$.
\end{theorem}

% \begin{proof}
% For the security of any query, we need to prove that the computation of every relational operator for a committed query transaction is valid.
% % The proof is based on two properties as follows. 
% First, according to the unforgeability of VSO under the q-SDH assumption~\cite{verifiable_set_1, verifiable_set_2}, the ADS for set operations guarantees that the computation of each cross-shard operator in delegates is valid.
% Next, when the proportion of malicious nodes in each shard is no more than the fault threshold $v$}, the safety of the intra-shard consensus holds~\cite{sok_sharding}, i.e., the honest nodes in each shard agree on the same valid block in each round. 
% Thus, for each of the other operators, since the nodes in the corresponding shards can validate the operators depending on their own storage, the intra-shard consensus can guarantee the computation is correct.
% Moreover, since each node stores the latest state of tables, the intra-shard consensus guarantees that the query result is computed from the latest state of tables.
% Note that the position indicator in the proof is used to support the computation of the operators between two cross-shard operators (or a cross-shard operator and the result).
% Therefore, the cross-shard mechanism can guarantee that the query transaction can be committed if and only if the processing of all relational operators in the query is valid.
% \end{proof}
\begin{proof}

We prove \textsc{Theorem 1} in three cases, corresponding to how \textsc{GriDB} defends against the three different adversaries in \textsc{Definition 1} for each cross-shard query in correctness, completeness, and freshness. 
We first describe the three cases:
% , while they all violate our assumptions.
\underline{Case 1}: This case means a tampered or fake row within the result is returned, which does not satisfy the query $q$. In this case, the tampered or fake row can pass the client's verification under the correctness in \textsc{Definition 1}. 
\underline{Case 2}: This case means a row that satisfies $q$ is missing from $R$. In this case, the incomplete result can pass the verification of the client under the completeness in \textsc{Definition 1}.
\underline{Case 3}: This case means the result $R$ involves an old row that satisfies $q$ but is not from the latest tables. In this case, the old result can pass the client's verification under the freshness in \textsc{Definition 1}.

If any of the above three cases occur, it means the computation of at least one relational operator (intra-shard or cross-shard operator) for a committed query is invalid, i.e., the malicious nodes in a related shard tamper with the executing of intra-shard operators during intra-shard consensus, or the main node generates a wrong result in the executing of cross-shard operators. 
However, this contradicts two assumptions. The first one is that when the proportion of malicious nodes in each shard is no more than the fault threshold $v$, the safety of the intra-shard consensus holds~\cite{sok_sharding}. Second, according to the unforgeability of VSO under the q-SDH assumption~\cite{verifiable_set_1, verifiable_set_2}, the ADS for set operations guarantees that the computation of each cross-shard operator in delegates is valid, and any invalid results can be detected by the intra-shard consensus.
\end{proof}

% We discuss several cases to show how \textsc{GriDB} defends against malicious nodes for each cross-shard query as follows. 
% 1) If a malicious delegate provides a wrong or old table, the honest nodes in the delegate's shard can detect it based on the query result, $acc(\cdot)$, and the position indicator in the proof and reject the query transaction during the cross-shard mechanism. 
% 2) If the main node is malicious and generates a wrong query result, the related shards for the query can detect it when computing all relational operators based on the proof. 
% 3) If there are some malicious nodes in a related shard, the intra-shard consensus can guarantee that they cannot tamper with the consensus result of the shard. 
% Thus, in these three cases, a query transaction including the invalid query results or fake proof can be detected and aborted.}

\textbf{Performance Analysis.} We analyze the time for a cross-shard query involving $m$ cross-shard operators as follows. 
Three steps occupy most of the time, i.e., the table transfer (Line 3), the proof generation (Line 4), and the confirmation latency of the query transaction (Line 6), which is also proved in \autoref{sec:experiment}. 
Thus, the analysis is developed around these three steps. 
For the table transfer, the time cost is linear to the size of the related tables.
% , which is a significant problem when the related tables are big. 
We introduce several refinements to reduce this time cost and improve query efficiency in \autoref{sec:query_efficiency}.
Next, according to~\cite{verifiable_set_1,verifiable_set_2}, the proof generation time for each set operation involving $N$ elements is $O(N\log^2N\log\log N)$. Thus, the proof generation time is $O(mN\log^2N\log\log N)$. Finally, the confirmation latency denotes the delay between the time that the query transaction is issued from the main node until the transaction is committed, which depends on the throughput, demand, and number of block confirmations of the blockchain.
%According to \autoref{sec:experiment}, the first part spends the most time. ...

\subsection{Inter-Shard Load Balancing}
\label{sec:live_migration}

\textbf{Motivation.} Observe that the drawback of the stop-restart approach in the strawman system results from the interruption for transaction processing during migration. Moreover, because the approach is on-chain, the migration occupies the transaction throughput of the shards involved, which interrupts the new transactions of the other tables.
% Assume that the blockchain database is going to migrate a table with 10000 rows from a shard to another shard and each shard can commit 400 rows on its blockchain in a second. The migration costs 25 seconds at least while all new transactions of the other tables in these two shards are halted. 
Thus, to avoid these drawbacks, we design an \emph{off-chain live migration} approach for \textsc{GriDB}.
Its main idea is to design an off-chain technique to minimize the number of on-chain transactions and a dual mode with cross-shard synchronization and concurrency control to minimize the impact of interruption to the migrating table during migration. 

\textbf{Design.} \autoref{fig:migration} illustrates the timeline of cross-shard migration and the messages exchanged between two shards. The life cycle of a table includes the following three modes. 

% \begin{algorithm}[t]
%     \SetKwInput{KwIn}{Input}
%     \SetKwInput{KwOut}{Output}
% 	\caption{Off-Chain Live Migration Protocol}
%     \label{alg:live_migration}
%         \KwIn{...}
%         \KwOut{...}
%         % \emph{Phase 1: Normal Mode}\\
%         % ...\\
%         A cross-shard control transaction $txn$ involving shards $\mathcal{S}$ and $\mathcal{D}$ is proposed \tcp{Init mode starts}
%         ...\\
%         ... \tcp{Dual mode starts}
% \end{algorithm}

\textit{1) Normal Mode:} The normal mode for a table  (called $\mathcal{T}$) is the period in which the data or query transactions of the table are processed normally by the shard it belongs to. 
% We call that the shard has full ownership for the table in this mode. 
The normal mode accounts for most of the time for the table. 

\begin{figure}[t]
    \centering
    \includegraphics[width=\linewidth]{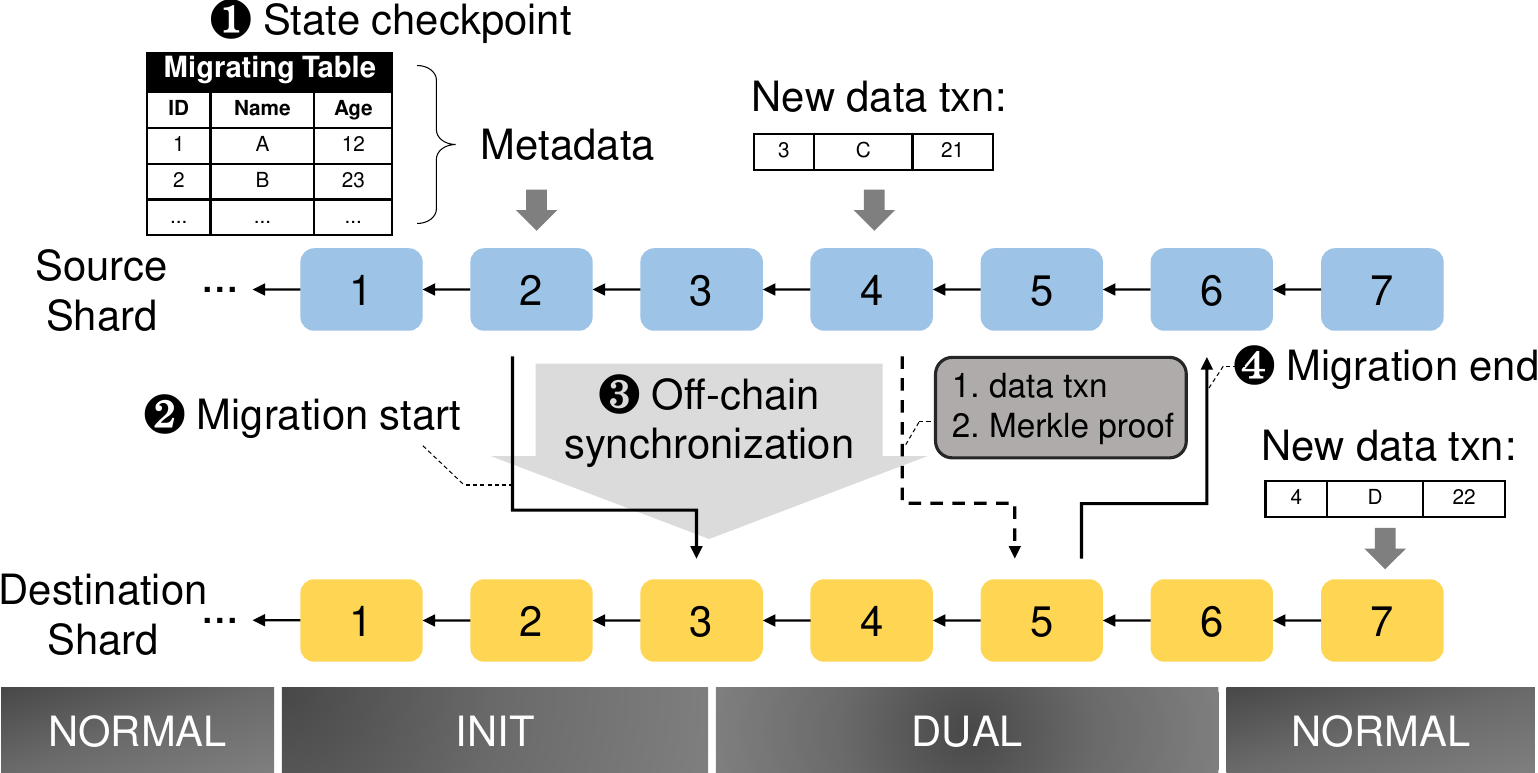}
    \caption{Overview of off-chain live migration. A solid line with arrowhead represents a cross-shard transaction and a dotted line with arrowhead represents an off-chain cross-shard communication.}
    \label{fig:migration}
\end{figure}

\textit{2) Init Mode:} When $\mathcal{T}$ is going to be migrated from the source shard (called $\mathcal{S}$) to the destination shard (called $\mathcal{D}$), the init mode starts. 
(We will discuss the trigger of table migration in \autoref{sec:scheduler} which aims for load balancing and guarantees that there is a super majority of honest nodes in $\mathcal{S}$ knowing $\mathcal{T}$ and $\mathcal{D}$.)
The nodes in $\mathcal{S}$ first construct the metadata for $\mathcal{T}$ via a hash function such as SHA (\autoref{fig:migration}-\ding{182}) and commit a cross-shard control transaction involving $\mathcal{S}$ and $\mathcal{D}$ (\autoref{fig:migration}-\ding{183}). 
The transaction includes the metadata and a block number, representing a checkpoint for $\mathcal{T}$ in this block number. 
When the control transaction is committed in both shards by the on-chain cross-shard mechanism, the init mode ends. 

\textit{3) Dual Mode:} In the dual mode, $\mathcal{S}$ begins to transmit $\mathcal{T}$ to $\mathcal{D}$. The transmission among shards is pluggable and can be implemented by one-to-one communication or gossip mechanisms (\autoref{fig:migration}-\ding{184}). The nodes in $\mathcal{D}$ only accept the table matching the metadata in the control transaction. Because the download of the whole table may cost a lot of time, we adopt a pre-copy scheme in which the nodes in $\mathcal{D}$ can pre-download $\mathcal{T}$ from $\mathcal{S}$ in the normal or init mode and validate it after the commitment of the control transaction. 

To keep the service for $\mathcal{T}$ during the dual mode, $\mathcal{S}$ continues to process the newcoming data and query transactions related to $\mathcal{T}$. The new data transactions in $\mathcal{S}$ may change the content of $\mathcal{T}$, thus $\mathcal{D}$ should be notified. It can be realized by committing all new data transactions as cross-shard transactions, however, which slows the service of $\mathcal{T}$ due to the overhead of cross-shard mechanism and blocks the throughput of $\mathcal{S}$ and $\mathcal{D}$ when the demand is high. Thus, we adopt an off-chain cross-shard notification mechanism based on Merkle tree as follows. First, in \textsc{GriDB}, similar to the other sharding systems \cite{monoxide,pyramid}, each node will be a light node for the other shards and store the block headers of all shards. It does not hurt the scalability, since the light nodes do not need to participate in the consensus and each header occupies little storage space and bandwidth. The notification is in the form of an off-chain message including a data transaction and its Merkle proof. 
Any nodes in $\mathcal{S}$ can notify $\mathcal{D}$ via gossip mechanism~\cite{gossip}. Based on the notification received, $\mathcal{D}$ gets the latest data transactions for $\mathcal{T}$.
For example, as shown in \autoref{fig:migration}, a new data transaction for the migrating table arrives in $\mathcal{S}$ and is committed at the 4-th block. Any honest node in $\mathcal{S}$ can send the new transaction with its Merkle proof in the 4-th block to $\mathcal{D}$ for synchronization of $\mathcal{T}$ between $\mathcal{S}$ and $\mathcal{D}$.
% both $\mathcal{S}$ and destination shard share the ownership of the migrating table, which means the data transactions related to the migrating table are processed in both shards in the form of cross-shard transactions. By such a way, although the processing for the newcoming data transactions related to the migrating table is slowed, it will not completely stop the service for the migrating table.

After a node in $\mathcal{D}$ completes downloading, it proposes a cross-shard control transaction involving $\mathcal{D}$ and $\mathcal{S}$ or participates in the consensus on the one proposed by another node to show that it has downloaded the table successfully. Thus, the transaction can be committed if the majority of honest nodes in $\mathcal{D}$ confirm that they have downloaded the migrating table (\autoref{fig:migration}-\ding{185}). We assume that the off-chain notification arrives reliably and without latency here, which will be discussed later. Finally, the migration is completed and $\mathcal{D}$ has full ownership of the migrating table. It means that the later transactions (e.g., data/query transactions and migration requests) for the migrating table are processed by $\mathcal{D}$ only.
% Besides, the nodes in $\mathcal{S}$ can remove the committed data transactions of $\mathcal{T}$ after migration since they do not need to process any new transactions of $\mathcal{T}$. ...}

\textbf{Asynchronous Issues.} In the above, we ignore some problems resulting from the network latency or malicious nodes. Thus, we discuss them and provide some designs as follows.

\textit{Problem 1:} In the init mode, due to the transaction latency existing in the blockchain, i.e., the delay between the time that a node sends a transaction to the network until the time that the transaction can be confirmed by all (honest) nodes, the metadata generated by different nodes may be in different versions. Thus, $\mathcal{S}$ may be unable to reach a consensus on the same control transaction.
To synchronize the metadata among nodes in $\mathcal{S}$, \textsc{GriDB} sets a rule as follows. 
When a node begins to generate the metadata, it stops processing any new data transactions of $\mathcal{T}$ and disagree on blocks including these transactions during consensus until the init mode ends. 
Note that a node still accepts the newly committed blocks and updates its local database state and the corresponding metadata even if it disagrees them.
Moreover, before a cross-shard control transaction is committed successfully, the nodes in $\mathcal{S}$ keep updating the metadata they generate based on the new block. If there is already the same control transaction proposed by other nodes waiting to be committed, the node can participate in its consensus. 

\textit{Problem 2:} In the dual mode, we adopt an off-chain notification mechanism to minimize the impact of interruption during migration. However, the off-chain communication among shards is not reliable thus the notifications may get lost. For the problem, we adopt the following designs. First, every new data transaction in the dual mode will be assigned an increasing sequence number before being committed in $\mathcal{S}$. Thus, if a node in $\mathcal{D}$ finds itself missing some transactions, it can directly request the corresponding notifications from the nodes in $\mathcal{S}$. Second, after the control transaction is committed (\autoref{fig:migration}-\ding{185}), $\mathcal{S}$ needs to commit a control transaction including the total number of new data transactions in the dual mode and sends the control transaction with a Merkle proof to $\mathcal{D}$. Besides, the nodes in $\mathcal{D}$ can actively ask the control transaction. Each node in $\mathcal{D}$ begins to process new transactions for $\mathcal{T}$ until it gets the total number of notifications and downloads all data transactions. Finally, $\mathcal{D}$ continues the service of $\mathcal{T}$ when the majority of honest nodes in $\mathcal{D}$ finish downloading.

\textbf{Security Analysis.} We first describe the formal definition of the security~\cite{zephyr} for our off-chain live migration as follows.

\begin{definition}
A migration is secure if achieving safety and liveness despite Byzantine failure. The safety requires serializable isolation, i.e., the migrating table's transactions run in serial order during migration, and durability, i.e., the committed transactions will not get lost after migration. The liveness indicates it eventually terminates. 
\end{definition}

\begin{theorem}
Our proposed off-chain live migration satisfies the security property as defined in Definition 2 if the proportion of malicious nodes in each shard is no more than the fault threshold $v$.
\end{theorem}

\begin{proof}
During the migration, only one of $\mathcal{S}$ and $\mathcal{D}$ has full ownership and processes transactions for the migrating table. The intra-shard consensus guarantees that there is a serializable order for transactions among nodes in each shard, thus achieving serializable isolation. For durability, in the init mode, because the honest nodes update their metadata before the control transaction is committed, $\mathcal{D}$ can download all data for $\mathcal{T}$ that are committed before the dual mode begins. The durability for transactions that are committed in the dual mode can be guaranteed by the final control transaction, including their total number in $\mathcal{S}$. Furthermore, the liveness can be guaranteed since the end conditions for each mode depend on the commitment of cross-shard transactions whose liveness has been shown to hold under the threat model of super-majority honest in each shard in any sharding works~\cite{sok_sharding}.
% In the dual mode, the cross-shard mechanism guarantees that the cross-shard transactions for the same table can be committed in both $\mathcal{S}$ and $\mathcal{D}$ in a serializable order.
\end{proof}

% We discuss several common cases to show how \textsc{GriDB} defends against malicious nodes for each inter-shard migration as follows. 
% 1)}
If the malicious nodes in $\mathcal{S}$ construct a wrong metadata for $\mathcal{T}$, the intra-shard consensus guarantees that the invalid transaction including the wrong metadata would not be committed. The nodes in $\mathcal{D}$ can detect wrong tables and wrong notifications according to the on-chain metadata and Merkle root, respectively. 
% 2) ...}

\textbf{Performance Analysis.} 
% We analyze the time cost for the migration as follows. 
The time for each migration is at least the latency of two cross-shard control transactions, one of which denotes the beginning of dual mode and the other the end. 
In parallel with the first one, $\mathcal{T}$ is transferred between shards. 
Its time depends on the adopted communication methods, network status, and network scale.
After the second one, to guarantee that all data transactions are received by $\mathcal{D}$, there is a control transaction in $\mathcal{S}$ and the nodes in $\mathcal{D}$ need to download the data transactions that they miss. The time of the step also depends on the network environment. In the worst case, if all data transactions during migration are missed by the majority of honest nodes, the nodes may need some time to download the transactions they missed. Besides, the service of $\mathcal{T}$ may be halted for a time due to the first and third designs for the asynchronous issues during migration.
% However, if the majority of honest nodes can download all data transactions during migration, the time of the step can be minimized.
% Although the rule may halt the service of the migrating table for a while, ... 
% For the storage cost, the nodes in $\mathcal{S}$ needs only needs the nodes in $\mathcal{D}$ in addition to the size of migrating table, the approach needs two cross-shard control transactions to denote the ... and the end of migration.
% Next, during the dual mode, to avoid the service halt of the migrating table, each data transaction is processed in the form of cross-shard transaction. 
% In most sharding blockchains, a cross-shard transaction requires more time to process than a single-shard transaction because it involves the intra-shard consensus of multiple shards. For example, the cross-shard transaction involving two shards take twice as long to be committed as single-shard transactions in~\cite{rapidchain, monoxide}.

\section{Design Refinement}
\label{sec:design_refinement}

\subsection{Cross-shard Query Efficiency}
\label{sec:query_efficiency}

Although the delegation-based approach in \autoref{sec:cross_shard_query} reduces the complexity for table transfer in a cross-shard query from $O(SN^2)$ (derived from the strawman in \autoref{sec:strawman}) to $O(S)$ where $N$ is the number of nodes in a shard and $S$ is the number of related shards for the query, transferring a huge table from the sub nodes to the main node still costs a lot. 
However, many rows are useless in practice. 
For example, for query \#5 in \autoref{sec:experiment_cross_shard} involving tables with millions of rows, the size of its final result only have single-digit items. 
\textsc{GriDB} optimizes the transferring of table among delegates as follows.

% To reduce the size of tables to be transferred among sub nodes and main nodes, 
We first optimize by applying the unary operators and binary operators involving tables in the same shard early in the tree of each cross-shard query.
Particularly, each sub node processes all selection operators related to its table and transfers the processed table to the main node. 
For example, in \autoref{fig:query_evaluation}, the selection operation $\sigma_{num} = 1$ is moved to the bottom of the tree and processed by the sub node in Shard A, reducing the size of Table 1 to be transferred to the main node in Shard B.
% Because the temporary outputs of these operations can be queried by one node itself, 
Because the execution order of the tree is bottom-up, the main node in Shard B downloads the part of tables including these temporary outputs, based on which it can continue to process the next operation. Moreover, some projection operators also can be applied early similar to the selection operators.
% Thus, the amount of data transferred between the sub nodes and main node can be reduced.
%adjusts the structure of the tree, 
%If an operation is a unary operation (e.g., selection or projection) like  \emph{log.oid} = 1, the main nodes mark it as \emph{on-the-fly} and download the part of tables which includes the temporary query outputs of these operations from the sub nodes. After receiving the temporary query outputs, the main nodes use the Pipelined Evaluation \cite[Chapter 13]{ramakrishnan1999database} to select an optimal   execution strategy by the relational algebra tree when applying the remaining selection to the temporary relation. 
% Bloom join~\cite{distributed2009ramesh} is another technique to avoid transferring unnecessary data in cross-shard query. In cross-shard query, some operations need one node to download the entire tables from other nodes and compare their items one by one, which causes a lot of waste of communication. 
%For example, for the cross-shard query evaluated in \autoref{sec:cross_shard_query} the join condition  \emph{ord.oid} =  \emph{log.oid} needs shard A to download the logistics table from shard B and compares them, but in fact a large amount of unmatched data is discarded after the join operation.

% Bloom join uses a bit-vector data structure of k size called bloom filter, and each bit of the bloom filter is set to 0 initially. 

Next, we adopt bloom filter (BF) for each cross-shard operator to filter out unnecessary data before transferring tables. 
BF~\cite{distributed2009ramesh} is a space-efficient probabilistic data structure used to test whether an element belongs to a set or not. In \textsc{GriDB}, before downloading the tables for a cross-shard operator, the main node can build a BF for the target column in its table and send the filter to the sub nodes. 
The sub nodes use it to filter their own table before transferring tables to the main node. 
Thus, most useless data are filtered out before being transmitted, reducing communication overhead.

%For all cross-shard related tables belonging to sub shards, the sub nodes compute every item in the join related row such as  \emph{log.oid} in \autoref{fig:query_evaluation} by a hash function into the range 0 to $k - 1$, and set the bit i to 1 if the calculated result is i. When the main node performs the \emph{join} operation, it uses the same hash function to build its own bloom filter of the related row such as  \emph{ord.oid}, and only downloads the items whose bit value is set to 1 at the same position in all related tables from other tables. Therefore, through the bloom join, most of the useless data is filtered before being transmitted, and the communication overhead is reduced.

\subsection{Load Balancing Scheduler}
\label{sec:scheduler}

A critical problem to achieve inter-shard balancing is how to generate a good planning strategy to distribute the load to shards and how to apply the strategy in a distributed and safe manner in \textsc{GriDB}. 
%For the traditional databases, there have been a lot of approaches to generate planning strategies~\cite{estore} based on the well-known ``bin packing'' algorithm. 

For a planning strategy, similar to distributed databases~\cite{squall, live_reconf_fast_tx}, \textsc{GriDB} follows a widely-used greedy planning algorithm~\cite{estore}.
% that operates in linear time. 
It iterates through the list of tables, starting with the one with the hottest demand. If the shard currently holding this table has a load exceeding the average demand, the algorithm migrates the table to the least load shard. 
The algorithm is easy to implement and has been proved to be efficient in many database scenarios~\cite{estore, squall, live_reconf_fast_tx}. 

To run the above algorithm in a decentralized manner, \textsc{GriDB} extends it by considering its execution in the existing sharding blockchains. Particularly, resharding phase is an important phase during the lifecycle of sharding blockchains~\cite{elastico, omniledger, rapidchain}. A sharding blockchain proceeds in epochs, where each epoch consists of a resharding phase followed by multiple intra-consensus rounds. 
During the resharding phase of an epoch, a shard will be elected as the reference shard based on a round-robin rule. 
% In the traditional system, the reference shard is responsible to check identity of each node.
In \textsc{GriDB}, the reference shard can act as the load balancing scheduler in the resharding phase. 
The leader of every shard computes the demand of each table and reports it to the reference shard in a cross-shard control transaction via the cross-shard mechanism.
After receiving the demand of every table, the leader of the reference shard proposes a cross-shard control transaction, involving all shards and including a new planning strategy in the following epoch, via the cross-shard mechanism. 
The cross-shard mechanism can guarantee that the new planning strategy is known by a majority of honest nodes in every shard.
Based on this strategy, the tables can be migrated among shards using the approach in \autoref{sec:live_migration}.

\textbf{Security Discussion.} The greedy planning algorithm and the table demand computing can be deterministic, thus any node can check the validity of their results.
This guarantees that only the cross-shard control transactions, including valid table demand or valid planning strategy, can be committed in all the related shards and the invalid ones will be aborted via the cross-shard mechanism.

\subsection{Cross-shard Insertion/Deletion/Update.} 

In \textsc{GriDB}, a data transaction can include a SQL statement for an insert, delete or update operation with nested subqueries or a multiple-table delete/update operation~\cite{mysql_ref}. If the nested subquery is cross-shard in the first case or the related tables belong to multiple shards in the second case, the data transaction involves multiple shards. 

For the first case, a data transaction including the cross-shard subquery result (or its hash) can be processed by the delegation-based approach in \autoref{sec:cross_shard_query} and committed as a cross-shard transaction. The transaction involves both the shard for the inserted/deleted/updated table and the related shards for the nested subquery. 
For the second case, a multi-table deletion/update can be considered as deleting/updating the specified rows in multiple tables based on a query to these related tables. Thus, it can also be processed as a cross-shard data transaction including query results similar to the first case.
Finally, because the cross-shard mechanism guarantees the atomicity of cross-shard transactions (see \autoref{sec:preliminary_sharding}), the cross-shard insertion, deletion, and update take effect in all related shards.
Any invalid cross-shard data transactions (e.g., including wrong subquery results) will be aborted by the cross-shard mechanism.
% The cross-shard mechanism in \textsc{GriDB} can guarantee that ...
% the deletion operation can be expressed by a binary operator $\setminus$ ... multi-table deletion ...
% To support such a cross-shard data transaction, \textsc{GriDB} adopts two corresponding refinements as follows. First, every insert, delete or update operation with nested subqueries can be decomposed into a collection of \emph{query blocks}~\cite{db_book}, each of which except the first one is a query with no nesting and exactly one \texttt{SELECT} keyword. We implement them by integrating the data and query transactions. ... Next, the deletion operation can be expressed by a binary operator $\setminus$ ...

% (how multi-shard operations handle the case in which part of a transaction fails (e.g, in one shard) while the other parts succeed. Some for of atomic commit is needed.)

\subsection{Horizontal/Vertical Table Partition}
\label{sec:fine_grained_sharding}

For each table, besides storing the entire table in one shard as discussed in \autoref{sec:system_overview}, \textsc{GriDB} can be developed into a fine-grained sharding blockchain database through \emph{horizontally} or \emph{vertically} partitioning the table into partitions. 
The former allows the table to be partitioned into disjoint sets of rows and the latter disjoint sets of columns. 
Load balancing can benefit from this fine-grained sharding for blockchain database since the database workload can be more evenly distributed to the blockchain shards. 

% We first discuss how a query to one partitioned table can be processed as follows. 
The partitions of each table are distributed to different shards, thus a table is stored in multiple shards. 
% For example, as shown in \autoref{fig:table_partition}, Table 1 is horizontally partitioned based on the range of $oid$ into two partitions stored in Shard A and B, respectively.
% Table 2 is vertically partitioned based on the column of $company$ into two partitions stored in Shard A and B, respectively.
For a horizontally partitioned table, each query needs to commit the same query transaction to all the related shards of the query. 
% For example, a query to Partition 1 of Table 1 needs to be committed in Shard A while a query to the entire Table 1 needs to be committed as the same two query transactions in Shard A and B, respectively.
% Moreover, a query to a vertically partitioned table can involve one or more shards. 
For a vertically partitioned table, each of its partitions can be regarded as an individual table. 
If a query involves the columns within a partition or the partitions related to the same shard, the query involves one shard and can be processed as a query transaction in the shard. 
However, if the query involves multiple columns of several partitions from different shards, it needs to be committed as a cross-shard query transaction. 
% We can process a query to the table being vertically partitioned by first joining the partitions.

% Next, consider a query joins two tables and one of them has been horizontally partitioned, we can achieve it by selecting a delegate for each partition. ...

% Similar to \autoref{sec:scheduler}, in the resharding phase, the load balancing scheduler can adopt some heuristics algorithms~\cite{integrating, automating} for a new partition strategy.

% \subsection{Public Key Management}

% According to \autoref{sec:preliminary_vso}, the size of public key equals ... To solve the problem, we propose a ... motivated by vChain$^+$~\cite{vchain_p}. ...}

% \begin{figure}[t]
%     \centering
%     \includegraphics[width=0.8\linewidth]{table partition.pdf}
% 	\caption{Illustration for table partition in \textsc{GriDB}.}
% 	\label{fig:table_partition}
% % \vspace{-16pt}
% \end{figure}

\section{Discussion}

\subsection{Permissioned and Permissionless Setting}

\textsc{GriDB} can be applied in both permissioned and permissionless scenarios, relying on the underlying blockchain sharding system. 
For a permissioned scenario, only a set of known, identified, but untrusted nodes can serve as blockchain nodes similar to the permissioned blockchain databases~\cite{blockchaindb,falcondb}. 
For a permissionless scenario, the blockchain database is public and open, and anyone can become a blockchain node without a specific identity. 
To resist Sybil attacks caused by the permissionless setting, \textsc{GriDB} can use a PoW-based identity generation as described in \autoref{sec:preliminary_sharding}, which is widely adopted by the permissionless blockchain sharding~\cite{omniledger,rapidchain}. 
Moreover, to compensate for the consensus overhead of blockchain nodes and avoid the Verifier Dilemma~\cite{10.1145/2810103.2813659}, \textsc{GriDB} will explicitly charge fee for each transaction and reward the blockchain nodes~\cite{10.14778/3329772.3329775}. 
We leave an incentive mechanism design for \textsc{GriDB} as our future work.

\subsection{General Join}

The cross-shard query authentication in \autoref{sec:cross_shard_query} works for equality join, because the cryptography primitive adopted in \textsc{GriDB} supports set intersection only.
For a general join case such as non-equijoin (i.e., join operation using comparison operator like $>$, $<$, $>=$, $<=$ with conditions), we can resort to cryptographic technologies with more general verifiable computing capacity, e.g., Trusted Execution Environment (TEE) and Succinct Arguments of Knowledge (SNARK), which will be left as our future works. 
% Motivated by existing TEE-based verifiable databases~\cite{veridb}, we can deploy TEE in the delegate of each shard and refine the cross-shard query authentication in \autoref{sec:cross_shard_query} as follows.
% For each cross-shard query, the intra-shard operators are securely executed inside an enclave of a delegate, while the cross-shard operations are executed in the main node's enclave, and TEE's secure channels can secure the cross-shard table transfer among the delegates.
% More detailed implementation and secure analysis will be left as our future works.

\section{Experimental Evaluation}
\label{sec:experiment}

\textbf{Implementation.} We implement a prototype of \textsc{GriDB} in Go~\cite{golang} based on Ethereum~\cite{geth} and Harmony~\cite{harmony}. 
We adopt a BFT consensus with BLS multi-signature~\cite{harmony_consensus} as the intra-shard consensus and a library named ate-pairing~\cite{atepairing} for the VSO. 
The on-chain cross-shard mechanism of \textsc{GriDB} is similar to that of Monoxide~\cite{monoxide}.
Particularly, to commit a cross-shard transaction, each of its related shards needs to validate and commit it.
Only if the transaction is committed in the blockchains of all its related shards, it is regarded as being committed successfully.
This can be checked based on a list of Merkle proofs, each corresponding to a related shard.
Besides, by checking the transaction hash included in every Merkle proof, it can be guaranteed that every related shard commits the same transaction.
The optimization designs in \autoref{sec:design_refinement} are also implemented.
To implement a MySQL interface to \textsc{GriDB}, we adopt a storage-agnostic SQL engine with in-memory table implementation~\cite{go_mysql}.
% and use CGO~\cite{cgo} to combine it into our Go framework.

% \begin{figure*}[t]
%     \centering
%     \begin{minipage}[t]{0.34\textwidth}
%         \centering
%         \includegraphics[width=1.05\linewidth]{experiment/tps.pdf}
%         \caption{Transaction throughput (cx means cross-shard ratio.)}
%         \label{fig:tps_shard}
%     \end{minipage}
%     \hfill
%     \begin{minipage}[t]{0.32\textwidth}
%         \centering
%         \includegraphics[width=1.05\linewidth]{experiment/tps.pdf}
%         \caption{...}
%         \label{fig:tps_new}
%     \end{minipage}
%     \hfill
%     \begin{minipage}[t]{0.3\textwidth}
%         \centering
%         \includegraphics[width=1.05\linewidth]{experiment/storage.pdf}
% 	\caption{Storage overhead per node.}
%     \label{fig:storage}
%     \end{minipage}
% \vspace{-16pt}
% \end{figure*}

\textbf{Setup.} The testbed is composed of 16 machines, each of which has an Intel E5-2680V4 CPU and 64 GB of RAM, and a 10 Gbps network link. Similar to~\cite{rapidchain, omniledger}, to simulate geographically-distributed nodes, we set the bandwidth of all connections between nodes to $20$ Mbps and impose a latency of $100$ ms on the links in our testbed.

\textbf{Baseline.} For comparison, we implement a non-sharding blockchain database. This type of blockchain database does not need to consider the challenge of cross-shard query and inter-shard balancing because each node stores and processes the whole database. For a fair comparison, this blockchain database also adopts the signature-based BFT consensus adopted by \textsc{GriDB} as its underlying consensus. 
The basic idea of the non-sharding blockchain database is similar to that of the existing works such as FalconDB and SEDBD~\cite{falcondb,SEBDB} except that they adopt the other variants of BFT consensus and support some other functionalities (such as indexes).
Moreover, we implement an on-chain sharding blockchain database including shard-cooperation cross-shard query and stop-restart inter-shard migration based on our strawman system in \autoref{sec:strawman}.
% 2) A sharding blockchain database same as \textsc{GriDB} except that it supports cross-shard query based on an ADS technique similar to FalconDB~\cite{falcondb}. We call it \textsc{GriDB}-$\alpha$. 3) A sharding blockchain database same as \textsc{GriDB} except that it adopts a stop-restart migration approach proposed in \autoref{sec:strawman} for the inter-shard balancing. We call it \textsc{GriDB}-$\beta$.

\textbf{Workloads.} We evaluate the performance of \textsc{GriDB} using TPC-H~\cite{tpch} which is widely used by the database community.
%and the previous blockchain databases. 
It consists of 8 tables for each dataset and 22 types of SQL queries.
% 9 tables and 5 types of transactions that simulate a warehouse-centric order-entry environment. 
Our experiments are run on a database with 16 TPC-H datasets which are uniformly split across shards. 
% ...}
Besides, we add data transactions, each of which insert, delete or update a new row, for the workload of each dataset.
% The cross-warehouse ratio is set to 1\% for the new-order transactions ...
% Each query transaction is composed of a random SQL query. 
To simulate the cross-shard query, there is a proportion of query transactions involving tables in different shards and the proportion is called \emph{cross-shard ratio}.
To simulate the workload imbalance, similar to~\cite{squall, live_reconf_fast_tx}, we set two imbalanced settings. For low imbalance, we adopt a Zipfian distribution where two-thirds of the accesses go to one-third of the datasets. For high imbalance, 40\% of transactions follow the Zipfian in low imbalance, and the other transactions target 4 datasets initially on the first shard.
%The second one is TPC-H\footnote{http://www.tpc.org/tpch/} which consists of 8 tables and 22 SQL queries.
%The Yahoo! Cloud Serving Benchmark (YCSB) which is a collection of workloads that are representative of large-scale services created by Internet-based companies.

\subsection{Overall Performance}

% \begin{figure}[t]
% 	\centering
% 	\subfloat[][Transaction throughput (cx means cross-shard ratio.)]{
% 		\begin{minipage}[t]{0.8\linewidth}
% 			\centering
% 			\includegraphics[width=\linewidth]{experiment/tps.pdf}
% 		\end{minipage}%
% 		\label{fig:tps_shard}
% 	}
	
% 	\subfloat[][Storage overhead per node.]{
% 		\begin{minipage}[t]{0.8\linewidth}
% 			\centering
% 			\includegraphics[width=\linewidth]{experiment/storage.pdf}
% 		\end{minipage}
% 		\label{fig:storage}
% 	}
% 	\caption{Scalability comparison of \textsc{GriDB} and the existing non-sharding blockchain database.}
% % \vspace{-14pt}
% \end{figure}

To evaluate the scalability, we measure the transaction throughput in TPS for the non-sharding blockchain database and \textsc{GriDB} with varying percentages of query transactions and cross-shard ratios. 
We deploy 30 nodes for each shard. 
\autoref{fig:tps_shard} shows that the measured TPS of \textsc{GriDB} increases linearly with the number of shards and decreases when there are more cross-shard query transactions in the workload. It is because the data transactions only involve one shard, and the verification is simple. 
However, a query transaction is computationally-intensive (it requires $0.17\sim2.38$ seconds even in a local database as discussed in \autoref{sec:experiment_cross_shard}) and needs the delegation-based procedure for cross-shard verification, thus, committing query transactions costs more.
Moreover, a query involving more shards causes more table transfers and more complex proof generation among the delegated nodes, which will be further studied in \autoref{sec:experiment_cross_shard}. 
In comparison with \textsc{GriDB}, the on-chain sharding blockchain database has a similar throughput when there are no cross-shard queries. 
However, its throughput drops to nearly $0$ when $50\%$ or $100\%$ of queries are cross-shard. 
It is because, for the on-chain one, the table transfer among shards caused by cross-shard queries can result in serious network blocked.

To evaluate the performance of \textsc{GriDB} for cross-shard data transactions, we pack cross-shard queries (used to delete the cross-shard query results) into cross-shard data transactions.
According to \autoref{fig:tps_shard}, \textsc{GriDB}'s throughput for cross-shard data transactions is similar to that for cross-shard query transactions.
It is because, as described in \autoref{sec:system_overview}, \textsc{GriDB} implements a delete statement by marking old data as invalid.
Except for reaching consensus on cross-shard query results like a cross-shard query transaction, a cross-shard data transaction needs to include the information of marking the results as invalid, and each node needs to delete the results from its in-memory tables. 
However, these additional overheads are negligible. 
Thus, the expense of cross-shard data transactions and cross-shard query transactions are similar.

\begin{figure}[t]
    \centering
    \includegraphics[width=0.82\linewidth]{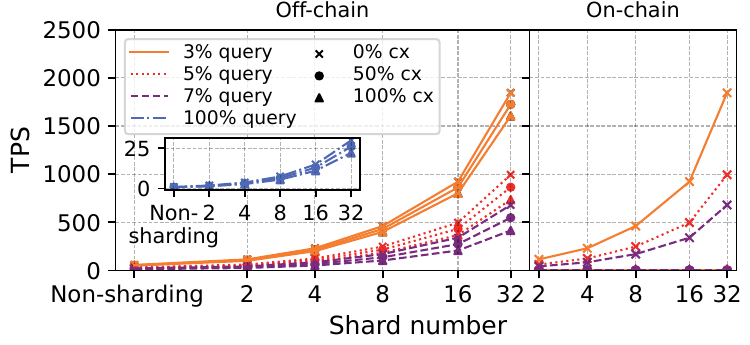}
    \caption{Transaction throughput for \textsc{GriDB}, the on-chain sharding blockchain database, and the non-sharding blockchain database (cx means cross-shard ratio.)}
    \label{fig:tps_shard}
\end{figure}

We also evaluate the storage overhead per node after loading all tables in the non-sharding blockchain database and \textsc{GriDB} with varying shard numbers. 
The results are given in \autoref{fig:storage}. 
Because each row is committed in the form of a data transaction and the data transactions are packed into blocks, loading the tables will introduce the block-related data including transaction-related and header-related data. 
From \autoref{fig:storage}, we can observe that, first, as the number of shards increases, the storage overhead for each node is reduced. 
Second, the transaction-related data cost half storage compared with the raw data.
Third, compared with the other data, the storage of headers can be ignored. 
Forth, because the largest table in the evaluation consists of 6 million rows, the public key size of verifiable set operation (VSO) is about $0.76$ GB. We regard the storage overhead caused by the public key as acceptable in the case of tables with millions of rows since it is considerably less than the recommended storage space of most blockchain nodes (such as 2 TB in Ethereum~\cite{geth_requirement}) nowadays.
Additionally, the storage of the on-chain sharding database is the same as that of \textsc{GriDB}.

\begin{figure}[t]
    \centering
    \includegraphics[width=0.75\linewidth]{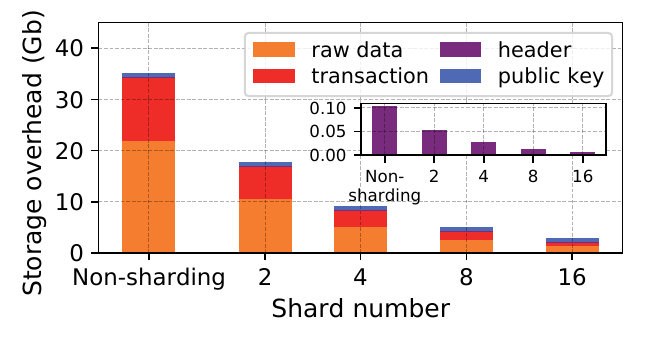}
	\caption{Storage overhead per node for \textsc{GriDB} and the non-sharding blockchain database.}
    \label{fig:storage}
% \vspace{-16pt}
\end{figure}

\subsection{Performance of Cross-shard Query}
\label{sec:experiment_cross_shard}

% \begin{figure}[t]
%     \centering
%     \includegraphics[width=\linewidth]{experiment/detail_query_time.pdf}
%     \caption{Time cost for each step for varying cross-shard queries.}
%     \label{fig:detail_query_time}
% \end{figure}

We evaluate the performance of cross-shard queries. For comparison, we adopt two approaches providing the same functionality as our cross-shard query.
% and has been adopted by a non-sharding blockchain database~\cite{falcondb}. 
These approaches are motivated by two previous works, i.e., vSQL~\cite{vsql} and libsnark~\cite{libsnark}, which can support arbitrary SQL queries based on interactive proof and SNARKs, respectively. Depending on either of these two works, any nodes can directly provide the result of a cross-shard query and a proof to the clients without consensus. We also evaluate the performance of the local computation for SQL in our nodes, based on MySQL. The server time is the time required for the server to evaluate the query and produce a valid proof and the client time is the time for the client to verify the proof. In \textsc{GriDB}, the server time is the duration from Line 2 to Line 6 in \autoref{alg:cross_shard_query}, and the client time is the duration of Line 7 in \autoref{alg:cross_shard_query}. 

\begin{table}[t]
	\centering
	\caption{Comparison of server and client times for evaluating queries using different approaches (The results for vSQL and SNARKs are provided in~\cite{vsql}.)}
	\scalebox{0.8}{
    \begin{tabular}{c|cc|cc|cc|c}
    \hline
    & \multicolumn{2}{c|}{vSQL} & \multicolumn{2}{c|}{SNARKs} & \multicolumn{2}{c|}{\textsc{GriDB}} & MySQL \\
    \hline\hline
    Query & Server & Client & Server & Client & Server & Client &  \\
    \hline
    \#19 & 4892s & 162ms & 196000s  & 6ms & 41.14s & 221ms & 2.38s \\
    \#6 & 3851s & 129ms & 19000s & 6ms & 4.93s & 221ms & 1.44s \\
    \#5 & 5069s & 398ms & 615000s & 110ms & 490.33s & 221ms & 1.95s \\
    \#2 & 2346s  & 508ms & 58000s & 40ms & 56.86s & 222ms & 0.17s \\
    \hline
    \end{tabular}
    }
    \label{table:cross_shard_query}
% \vspace{-16pt}
\end{table}

As a representative example, we pick the query \#19, \#6, \#5, \#2 in TPC-H and the results are given in \autoref{table:cross_shard_query}. These queries include most SQL types, e.g., join, range, min and nested query. According to \autoref{table:cross_shard_query}, the server time of \textsc{GriDB} is orders of magnitude less than that of vSQL and SNARKs while the client time is similar. For the server time, it is because our cross-shard query only constructs the expensive ADS for a few cross-shard operators while the security of the other operations depends on the intra-shard consensus. For the client time, it is because the clients of \textsc{GriDB} only need to check whether their query transactions are confirmed or not via SPV. Note that the evaluation is based on the worst case, which means the tables for each cross-shard query are all located in different shards. 

The time cost of each step for the queries in \textsc{GriDB} is summarized in \autoref{table:detail_query_time}. The results shows that the three steps occupy most of the time, matching the performance analysis in \autoref{sec:cross_shard_query}. Furthermore, from \autoref{table:detail_query_time}, we have the following observations. First, according to the result of MySQL in \autoref{table:cross_shard_query}, query \#19 is the most complex one and the nodes spend more time on validating it during the intra-shard consensus, thus its confirmation latency is the most. Then, the proof generation and table transfer of query \#5 is the most, because the query needs to join six tables, which results in six cross-shard operators in the worst case. Finally, the time cost of query \#6 is the least, because it is a simple 3-dimensional range query followed by an aggregation for a single table.

\begin{table}[t]
	\centering
	\caption{Time of each step for queries in \textsc{GriDB} (CL: Confirmation latency, PG: Proof generation, TT: Table transfer.) 
% 	The number in the parenthesis denotes the number of cross-shard operators for a query in the worst case.
	}
	\scalebox{0.8}{
    \begin{tabular}{c|c|c|c|c}
    \hline 
    Query & CL & PG & TT & The others \\
    \hline\hline
    \#19 & 4.38s & 36.74s & 4.04ms & 10ms \\
    \#6 & 3.44s & 1.44s & 0s & 5ms \\
    \#5 & 3.95s & 483.01s & 3.2s & 100ms \\
    \#2 & 2.17s & 54.46s & 139.57ms & 80ms \\
    \hline
    \end{tabular}
	\label{table:detail_query_time}}
\end{table}

We also evaluate the performance of the cross-shard query of \textsc{GriDB} with varying table size and number of related shards. We scale the number of rows in the largest participating table in query \#5 from $6 \times 10^3$ to $6 \times 10^6$ and distribute its participating tables to $1\sim6$ shards. \autoref{fig:query_different_table_size} and \autoref{fig:query_different_shard_num} show that the time cost is significantly reduced when the participating tables are smaller or there are fewer related shards. It is because the complexity of proof generation depends on the participating table size and the number of cross-shard operators, matching the analysis in \autoref{sec:cross_shard_query}.
% Q5 is a query that joins 6 tables 

\begin{figure}[t]
	\centering
	\subfloat[][Table size]{
		\begin{minipage}[t]{0.5\linewidth}
			\centering
			\includegraphics[width=0.95\linewidth]{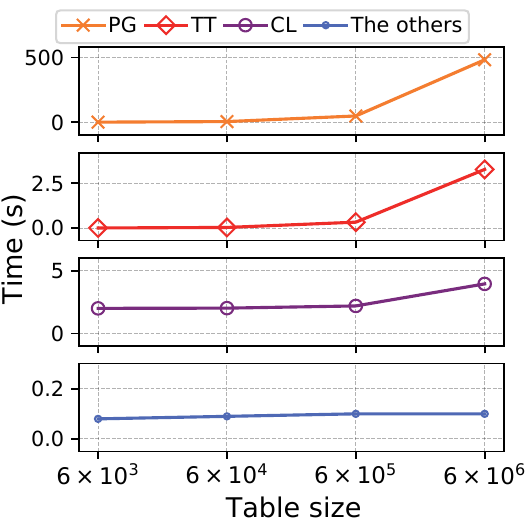}
		\end{minipage}%
		\label{fig:query_different_table_size}
	}
	\subfloat[][Number of the related shards]{
		\begin{minipage}[t]{0.5\linewidth}
			\centering
			\includegraphics[width=0.95\linewidth]{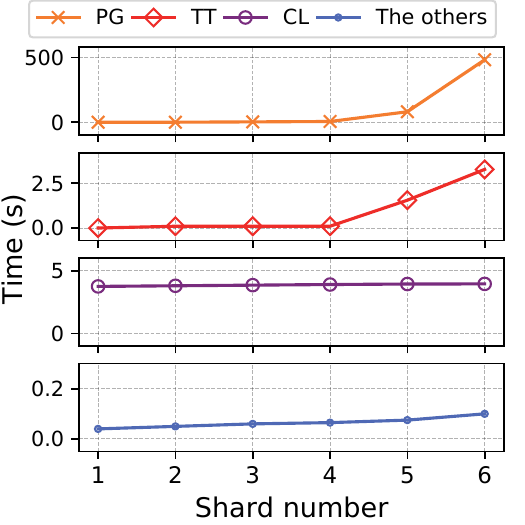}
		\end{minipage}
		\label{fig:query_different_shard_num}
	}
	\caption{Performance for query \#5 with different table size and number of related shards in \textsc{GriDB}.}
	\label{fig:query_different_factor}
% \vspace{-16pt}
\end{figure}

\subsection{Performance of Inter-shard Balancing}

\begin{figure}[t]
	\centering
	\subfloat[][Fluctuation of throughput during migration.]{
		\begin{minipage}[t]{\linewidth}
			\centering
			\includegraphics[width=\linewidth]{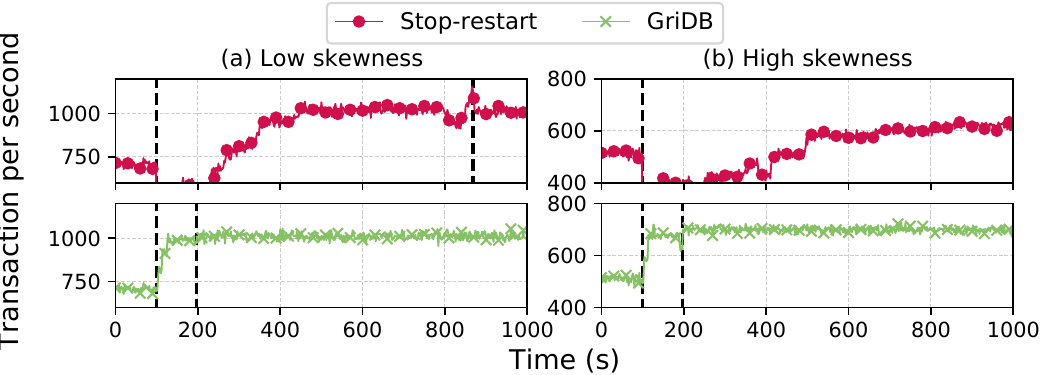}
		\end{minipage}%
		\label{fig:live_migration_period_1}
	}
	\qquad
	\subfloat[][Statistics on migration time and throughput.]{
	\scalebox{0.8}{
        \begin{tabular}{c|c|c}
        \hline 
        Migration time (s), TPS & Low skewness & High skewness \\
        \hline\hline
        Stop-restart & 770, 981 & 10678, 623  \\
        \textsc{GriDB} & 96, 1012 & 96, 700 \\
        \hline
        \end{tabular}}
        }
		\label{fig:live_migration_period_2}
    \caption{Transaction throughput during inter-shard migration with varying skewness.}
    \label{fig:migration_period}
\end{figure}

\begin{figure}[t]
		\centering
		\subfloat[][Migration time]{
			\begin{minipage}[t]{0.3\linewidth}
				\centering
				\includegraphics[width=\linewidth]{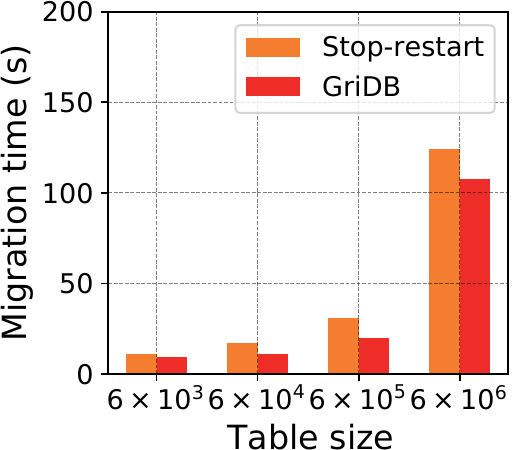}
			\end{minipage}%
			\label{fig:different_table_size_migration_time}
		}
		\hfill
		\subfloat[][Latency of txs for the migrating table]{
			\begin{minipage}[t]{0.3\linewidth}
				\centering
				\includegraphics[width=\linewidth]{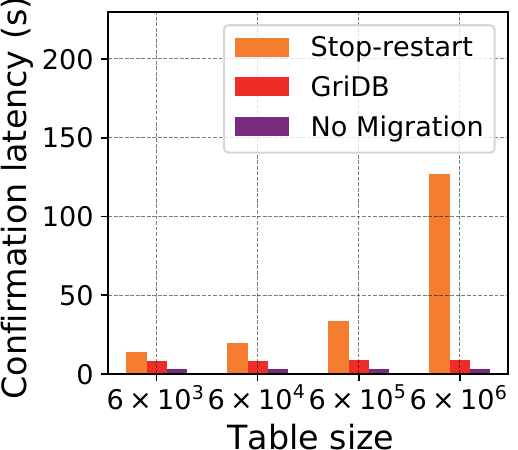}
			\end{minipage}
			\label{fig:different_table_size_related_confirmation_latency}
		}
		\hfill
		\subfloat[][Latency of txs for the other tables]{
			\begin{minipage}[t]{0.3\linewidth}
				\centering
				\includegraphics[width=\linewidth]{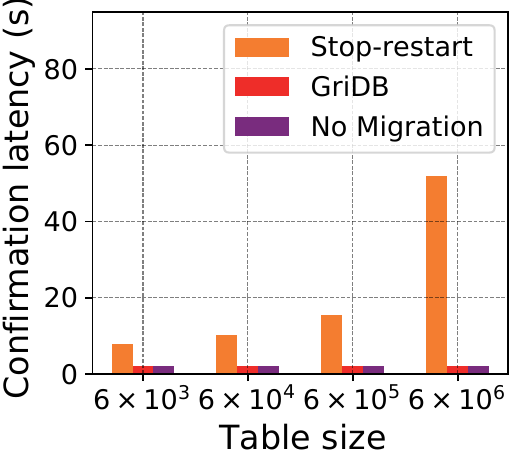}
			\end{minipage}%
			\label{fig:different_table_size_confirmation_latency}
		}
		\caption{Inter-shard migration for tables with varying size.}
		\label{fig:detailed_live_migration}
% \vspace{-5pt}
\end{figure}

% \begin{figure}[t]
%     \centering
%     \includegraphics[width=\linewidth]{experiment/live_migration_period.pdf}
%     \caption{Transaction throughput during inter-shard migration with varying skewness.}
%     \label{fig:migration_period}
% \vspace{-14pt}
% \end{figure}

% \begin{figure}[t]
%     \centering
%     \includegraphics[width=0.8\linewidth]{experiment/live_migration_period_low.pdf}
%     \caption{The throughput timeline for low skewness}
%     \label{fig:live_migration_period_low}
% \end{figure}

% \begin{figure}[t]
%     \centering
%     \includegraphics[width=0.8\linewidth]{experiment/live_migration_period_high.pdf}
%     \caption{The throughput timeline for high skewness}
%     \label{fig:live_migration_period_high}
% \end{figure}

We evaluate the throughput during migration via the off-chain live migration in \textsc{GriDB} and the stop-restart approach in the on-chain sharding blockchain database with various skewed workloads and the results are given in \autoref{fig:migration_period}. The process includes 48 migrations. After migration, the throughput increases by $1.40\times$ for the low skewness and $1.37\times$ for the high skewness.
It shows the load balancing among shards is helpful for the performance of sharding blockchain database. The off-chain live migration can shorten the migration time by nearly 87\% compared with the stop-restart approach for the low skewness and 99\% for the high skewness. 
Furthermore, the performance degradation in \textsc{GriDB} is minimal during migration. 
It is because, in \textsc{GriDB}, the off-chain manner significantly reduces the number of on-chain transactions, avoiding the massive overhead for consensus, and the dual mode minimizes service interruption during migration using the cross-shard off-chain notification.

\autoref{fig:detailed_live_migration} plots the impact of the table size on the migration time, the confirmation latency of transactions for the migrating table and the other tables in the shards involved. 
\autoref{fig:different_table_size_migration_time} shows that it costs more time to migrate a bigger table for both approaches. 
However, the migration time in \textsc{GriDB} is less than that in the stop-restart approach because there are only two on-chain transactions in \textsc{GriDB}, and a bigger table only requires more transmission time rather than more consensus rounds like the stop-restart approach. 
% In comparison, for stop-restart approach, the number of on-chain transactions equals the table size. ...
According to \autoref{fig:different_table_size_related_confirmation_latency} and \autoref{fig:different_table_size_confirmation_latency}, in \textsc{GriDB}, the confirmation latency for the tables during the migration is similar to that during normal mode (i.e., ``No Migration'' in the figures). 
Furthermore, the latency of transactions in the migrating table during migration is more than the latency during normal mode. 
It is because, in the dual mode, they are required to notify the destination shard.

% We randomly distribute the workload among tables based on a distribution of ...

\section{Related Work}
\label{sec:related_work}

\subsection{Blockchain Database}

We summarize the recent works that support various queries (e.g., SQL, key-value query, semantic query) on the blockchain as follows.

In the conventional blockchain, 
% because blocks are connected by hash pointers, 
a query requires nodes to traverse the whole blockchain to guarantee a complete search.
To avoid the expensive traversing, Pei \textit{et al.} design a Merkle Semantic Trie for efficient semantic query~\cite{semantic}.
SEBDB~\cite{SEBDB} builds three indices for efficient SQL queries, which can be complementary to our work. For instance, one can add these indices into each node in \textsc{GriDB} to improve the query efficiency.
% Nathan \textit{et al.} ...~\cite{relational}.
Motivated by outsourced databases~\cite{survey_outsourced}, some works provide verifiable queries with blockchain clients, enabling clients to verify query results from untrusted nodes.
SEBDB~\cite{SEBDB} designs a Merkle B-tree-based ADS for verifiable range query. 
Zhang \textit{et al.} support verifiable Boolean range query~\cite{vchain,gemtree} while
Falcon~\cite{falcondb} supports verifiable SQL~\cite{integridb} for blockchain databases.
However, they are constructed on top of the non-scale-out blockchains suffering from poor scalability.
%Furthermore, to penalize the malicious servers, Falcon employs a smart contract-empowered incentive mechanism allowing clients to challenge the query results from the malicious servers. 

The work most related to us is BlockchainDB~\cite{blockchaindb}, a key-value database on top of a sharding blockchain. In BlockchainDB, each transaction includes a key-value pair and the database layer provides clients with the interfaces of get, put, and verify operations. 
% The processing of each operation is distributed to the shard according to which table the operation belongs to. 
However, due to the simplicity of the key-value data model, every shard is isolated and every operation can be served by a single shard.
In comparison, \textsc{GriDB} considers a relational sharding blockchain database, thus having richer transactional semantics than BlockchainDB. 
It brings the challenge of cross-shard queries as discussed in \autoref{sec:challenges}.
Moreover, the data management for sharding, i.e., the inter-shard balancing, has not been considered in BlockchainDB. 
Besides, \textsc{GriDB}'s off-chain live migration can be applied to BlockchainDB by considering a key-value database as a particular case of a relational database with two-column tables.
% The importance of the balancing has been proved by many related works~\cite{live_reconf_fast_tx, squall, zephyr,estore}.

\subsection{Sharding}

Elastico~\cite{elastico} is the first decentralized sharding blockchain. 
% Before participating in the sharding system, each node solves a PoW puzzle. 
% Based on the least-significant bits of the PoW solution, the nodes are divided into different shards.
Every shard is responsible for validating a set of transactions via PBFT. 
A final shard verifies all the transactions received from shards and pack them into a global block. 
However, it only realizes verification sharding and each node needs to store all blocks. 
Omniledger~\cite{omniledger} is the first blockchain achieving full sharding by a client-driven cross-shard mechanism. 
% In particular, to commit an Unspent Transaction Output (UTXO)-based transaction, a client first asks for proofs from all input shards and then sends these proofs to all output shards. 
% If an output shard receives all proofs, it commits the transaction.
% Otherwise, the commitment of transactions will be failed, and other shards will roll back the transaction.
Another client-driven sharding system named Chainspace~\cite{chainspace} is presented to support sharding for smart contracts. 
The client-driven cross-shard mechanisms put extra burden on typically lightweight user nodes and are vulnerable to denial-of-service attacks. To further improve the performance, researchers propose shard-driven mechanisms, e.g., RapidChain~\cite{rapidchain} and
% For example, RapidChain~\cite{rapidchain} proposes a transfer mechanism for UTXO. In particular, for each cross-shard transaction, the input shards first transfer all involved UTXOs to the output shard by sub-transactions. Then, the cross-shard transaction can be transformed into a single-shard transaction and processed in the output shard.
Monoxide~\cite{monoxide}.
% proposes a relay mechanism for account/balance transactions. 
% Each cross-shard transaction will be divided into a sequence of sub-transactions. Each sub-transaction includes the operations involved accounts in one shard. An additional relay transaction is used as a cross-shard message when processing from a sub-transaction to another sub-transaction, i.e., from a shard to another shard.
Additionally, some works~\cite{pyramid,smart_contract_sharding,sigmod_sharding} share a similar idea in which some nodes store the blockchain state of multiple shards and thus can efficiently validate and execute cross-shard transactions.

Most of the existing sharding blockchains focus on transfer transactions (i.e., user-to-user transfers of digital funds) and smart contracts.
% and only provide some simple query interfaces based on the hash or number of transactions or blocks. 
Thus, their challenge of cross-shard transactions results from cross-shard payments (i.e., the transfer between accounts in different shards) or cross-shard smart contract calls (i.e., a smart contract that utilizes smart contracts at different shards).
In comparison, in \textsc{GriDB}, the challenge result from \emph{cross-shard queries} (i.e., queries to tables from different shards), \emph{cross-shard data operations} (i.e., insert, delete or update operations on tables from different shards), and \emph{inter-shard load balancing} caused by workload imbalance among shards. 
As discussed in \autoref{sec:introduction}, these cross-shard database transactions make the on-chain cross-shard mechanism of the existing blockchain sharding systems suffer from extremely poor performance.
% They do not consider the transactions which include semantic information. 
% Thus, they cannot guarantee the correctness, completeness, freshness, and availability of relational queries requested by the clients.

\subsection{Off-chain}

Off-chain protocols aim for blockchain scalability and build on top of the existing blockchains without changing trust assumptions~\cite{sok_offchain}. 
The main idea is to avoid processing every transaction via the consensus and instead utilize the consensus only to undertake critical tasks (e.g., settlement and dispute resolution)\footnote{Note that the word ``off-chain'' in BlockchainDB~\cite{blockchaindb} means it enables a node of a shard to verify the data from the other shards, which is different from our paper.}.
There are two types of off-chain protocols.
The first one is \emph{channels} in which the involved parties can update their balance unanimously and privately by exchanging authenticated transitions off-chain~\cite{lightning, raiden, sprites, cycle}.
The second one is \emph{commit-chains} which deploys a centralized but untrusted party to collect transactions on a child-blockchain and periodically update them to the parent-blockchain~\cite{khalil2018commit,poon2017plasma}.

\textsc{GriDB} shares the same idea with the existing off-chain protocols but uses a different approach (i.e., ADS) and targets a different problem (i.e., the high expense of cross-shard mechanism in the sharding blockchain database). 
To the best of our knowledge, \textsc{GriDB} is the first work to present an off-chain cross-shard mechanism to ease massive cross-shard data exchange, which may motivate the design of other sharding blockchain databases in the future.
% with different data models or query languages.

\section{Conclusion}
\label{sec:conclusion}

We present \textsc{GriDB}, a sharding blockchain database that achieves a few thousand transactions per second on about one thousand nodes in a Byzantine environment while supporting the functionalities of data insert/update, relational queries, and database management. 
The off-chain cross-shard mechanism, including delegation-based cross-shard query and off-chain live migration, is the key contribution in \textsc{GriDB}. They offer a database layer of abstraction on top of the existing sharding blockchain and hide the complexity of the data and workload partition in the underlying sharding blockchain from the clients. \textsc{GriDB} also includes some database key components, including query optimization, and load scheduler. We plan to study other sharding strategies in future work, such as functional partitioning for sharding blockchain databases.

\begin{acks}
This research was supported by fundings from the Key-Area Research and Development Program of Guangdong Province under grant No. 2021B0101400003,  Hong Kong RGC Research Impact Fund (RIF) with the Project No. R5060-19, General Research Fund (GRF) with the Project No. 152221/19E, 152203/20E, 152244/21E, and 152169/22E, and the National Natural Science Foundation of China (NSFC) 61872310, 62172453 and 62272496, Shenzhen Science and Technology Innovation Commission (JCYJ20200109142008673), the Major Key Project of PCL (PCL2021A06), and the Pearl River Talent Recruitment Program (2019QN01X130).
We thank all anonymous reviewers who helped improve the paper.
\end{acks}

\bibliographystyle{ACM-Reference-Format}
\bibliography{jobname}

\end{document}